\documentclass[10pt,preprint]{sigplanconf-pldi16}

\usepackage[T1]{fontenc}
\usepackage{flushend}
\usepackage[scaled=0.75]{beramono}
\usepackage[colorlinks=true,allcolors=blue,breaklinks,draft=false]{hyperref}
\usepackage[pdftex]{graphicx}
\usepackage{amsmath,amssymb,stmaryrd}
\usepackage[inference]{semantic}
\usepackage{nanoml}
\usepackage{xspace}
\usepackage{paralist}

\usepackage{amsthm}
\theoremstyle{definition}

\newtheorem*{example*}{Example}
\newtheorem{theorem}{Theorem}
\newtheorem{lemma}{Lemma}

\newcommand{\custompar}[1]{\parskip 0pt \textbf{\textit{#1}}}

\newcommand{\Implies}{\Rightarrow}
\renewcommand{\And}{\wedge}

\newcommand{\Subt}{<:}
\newcommand{\Consi}{\mathrel{\mathrm{\wedge:}}}
\newcommand{\HasT}{\;\mathrel{\mathrm{::}}\;}
\newcommand{\App}[2]{{#1}\ {#2}}

\newcommand{\env}{\Gamma}
\newcommand{\produce}{\;\uparrow\;}
\newcommand{\consume}{\;\downarrow\;}
\newcommand{\produceB}{\;{\mathrm{\uparrow_b}}\;}
\newcommand{\consumeB}{\;{\mathrm{\downarrow_b}}\;}
\newcommand{\funT}[3]{{#1}\colon {#2} \to {#3}}
\newcommand{\before}{\prec}
\newcommand{\ttop}{\mathsf{top}}
\newcommand{\tbot}{\mathsf{bot}}
\newcommand{\refine}{\mathsf{Refine}}
\newcommand{\fresh}{\mathsf{Fresh}}
\newcommand{\solve}{\mathsf{Solve}}
\newcommand{\horn}{\mathsf{Horn}}
\newcommand{\strengthen}{\mathsf{Strengthen}}
\newcommand{\dom}{\mathsf{dom}}
\newcommand{\fv}{\mathsf{FV}}
\newcommand{\pathc}[1]{\mathsf{P}({#1})}
\newcommand{\bnd}[2]{\mathsf{B}_{#1}({#2})}
\newcommand{\contT}[2]{\T{let}\ {#1}\ \T{in}\ {#2}}
\newcommand{\quals}{\mathbb{Q}}
\newcommand{\entailsQ}{\vdash_{\quals}}
\newcommand{\constraints}{\mathcal{C}}
\newcommand{\horns}{\mathcal{H}}
\newcommand{\typeass}{\mathcal{T}}
\newcommand{\liquidass}{\mathcal{L}}
\newcommand{\valuation}[2]{\llbracket{#1}\rrbracket_{#2}}

\newcommand{\ie}{\textit{i.e.}\@\xspace}

\newcommand{\lang}{\textsc{Synquid}\xspace}
\newcommand{\tool}{\textsc{Synquid}\xspace}
\newcommand{\exCount}{64\xspace}

\usepackage{multirow}
\usepackage{hhline}
\newcommand{\head}[1]{\emph{#1}}

\makeatletter
\protected\def\ccell#1#{%
  \kern-\fboxsep
  \@ccell{#1}%
}
\def\@ccell#1#2#3{%
  \colorbox#1{#2}{#3}%
  \kern-\fboxsep
}
\makeatother

\usepackage[all=normal,wordspacing]{savetrees}

\newif\iflong
\longtrue
\begin{document}
% \toappear{}

\setlength{\pdfpageheight}{\paperheight}
\setlength{\pdfpagewidth}{\paperwidth}

\title{Program Synthesis from Polymorphic Refinement Types}
\iflong
\subtitle{Extended Version}
\fi

\authorinfo{Nadia Polikarpova\and Ivan Kuraj\and Armando Solar-Lezama}
           {MIT CSAIL, USA}
           {\{polikarn,ivanko,asolar\}@csail.mit.edu}

\maketitle

\begin{abstract}
We present a method for synthesizing recursive functions that provably satisfy a given specification in the form of a polymorphic refinement type.
We observe that such specifications are particularly suitable for program synthesis for two reasons.
First, they offer a unique combination of expressive power and decidability,
which enables automatic verification---and hence synthesis---of nontrivial programs.
Second, a type-based specification for a program can often be effectively decomposed into independent specifications for its components,
causing the synthesizer to consider fewer component combinations
and leading to a combinatorial reduction in the size of the search space.
At the core of our synthesis procedure is a new algorithm for refinement type checking,
which supports specification decomposition. 

We have evaluated our prototype implementation on a large set of synthesis problems and found 
that it exceeds the state of the art in terms of both scalability and usability. 
The tool was able to synthesize more complex programs than those reported in prior work 
(several sorting algorithms and operations on balanced search trees),
as well as most of the benchmarks tackled by existing synthesizers, 
often starting from a more concise and intuitive user input. 
\end{abstract}

% \category{F.3.1}{Logics and Meanings of Programs}{Specifying and Verifying and Reasoning about Programs}
% \category{I.2.2}{Automatic Programming}{Program Synthesis}

% \terms
% Languages, Verification

\keywords
Program Synthesis, Functional Programming, Refinement Types, Predicate Abstraction

\section{Introduction}\label{sec:intro}

The key to scalable program synthesis is modular verification.
Modularity enables the synthesizer to prune candidates for different subprograms independently,
whereby combinatorially reducing the size of the search space it has to consider.
This explains the recent success of \emph{type-directed} approaches to synthesis of functional programs~\cite{OseraZd15,FeserChDi15,FrankleOWZ16,GveroKuKuPi13}:
not only do ill-typed programs vastly outnumber well-typed ones,
but more importantly, a type error can be detected long before the whole program is put together.

Simple, coarse-grained types alone are, however, rarely sufficient to precisely describe a synthesis goal.
Therefore, existing approaches supplement type information with other kinds of specifications,
such as input-output examples~\cite{AlbarghouthiGuKi13,OseraZd15,FeserChDi15},
or pre- and post-conditions~\cite{LeinoMi12,KneussKuKuSu13}.
Alas, the corresponding verification procedures rarely enjoy the same level of modularity as type checking,
thus fundamentally limiting the scalability of these techniques.

In this work we present a novel system that pushes the idea of type-directed synthesis one step further by taking advantage of \emph{refinement types}~\cite{Flanagan06,RondonKaJh08}: 
types decorated with predicates from a decidable logic.
For example, imagine that a user intends to synthesize the function \T{replicate},
which, given a natural number $n$ and a value $x$, produces a list that contains $n$ copies of $x$.
In our system, the user can express this intent by providing the following signature:
$$
\T{replicate}\HasT n:\T{Nat}\ \to x:\alpha \to \{\nu\colon\T{List}\ \alpha \mid \T{len}\ \nu = n\}
$$
Here, the return type is refined with the predicate $\T{len}\ \nu = n$,
which restricts the length of the output list to be equal to the argument $n$;
\T{Nat} is a shortcut for $\{\nu\colon\T{Int}\mid\nu \geq 0\}$, the type of integers that are greater or equal to zero%
\footnote{Hereafter the bound variable of the refinement is always called $\nu$ and the binding is omitted.}.
Given this signature,
together with the definition of \T{List} and a standard set of integer components (which include zero, decrement function, and inequalities),
our system produces a provably correct implementation of \T{replicate}, shown in \autoref{fig:replicate}, within fractions of a second.

\begin{figure}
\begin{nanoml}
replicate :: n:Nat -> x:_a -> {List _a |len _v == n}
replicate = \n.\x.if n <= 0
  then Nil
  else Cons x (replicate (dec n) x)
\end{nanoml}
\caption{Refinement type signature of \T{replicate} and the code synthesized from this signature.}\label{fig:replicate}
\end{figure}

We argue that refinement types offer the user a convenient interface to a program synthesizer:
the signature above is only marginally more complex than a conventional ML or Haskell type.
Contrast that with example-based synthesis,
which would require a conventional type together with multiple input-output pairs,
and in return provide much weaker correctness guarantees.

The \T{replicate} example is a perfect illustration of the power of parametric polymorphism for specifying program behavior.
Even though the signature in \autoref{fig:replicate} never says explicitly that each element of the output list must equal $x$,
it nevertheless captures the semantics of \T{replicate} completely:
since the function knows nothing about the type parameter $\alpha$,
it has no way of constructing any values of this type other than $x$.
This surprising expressiveness of polymorphic types had been long known~\cite{Wadler89},
but combined with refinements, it enables full-fledged higher-order reasoning within the type system:
a caller of \T{replicate} can instantiate $\alpha$ with any refinement type,
obtaining the fact that whenever $x$ has a certain property, every element of the output list shares that same property.

Perhaps surprisingly, prior work on \emph{liquid types}~\cite{RondonKaJh08,KawaguchiRoJh09,VazouRoJh13,VazouSeJh14} has shown that
this type of higher-order reasoning can be fully automated for a large class of programs and properties.
The \emph{liquid type inference} algorithm~\cite{RondonKaJh08}
uses a combination of Hindley-Milner unification and least-fixpoint Horn clause solver based on predicate abstraction
to discover refined instantiations for polymorphic types 
and ultimately reduce verification to proving quantifier-free formulas over simple refinement predicates,
efficiently decidable by SMT solvers.
The unique combination of expressive power and decidability offered by polymorphic refinement types
makes them ideal for program synthesis.

\custompar{Technical Challenges.}
Unfortunately, liquid type inference cannot be applied out of the box to the context of synthesis.
Designed for the setting where, given a program, the goal is to construct its type,
the inference algorithm starts from the leaves of the program, whose types are known,
and propagates type information bottom-up, 
constructing types of terms from the types of their subterms.
In program synthesis, however, the setting is different:
here the top-level type is given, and the goal is to construct the program.
One way to do so is to exhaustively explore program candidates,
performing liquid type inference on each one,
and then checking if the inferred type matches the given specification.
While this approach does rule out many ill-typed partial programs,
it fails to take advantage of the specification for guiding the search.
A more promising approach would propagate type information \emph{top-down} from the specification, 
using it to filter out irrelevant partial solutions.

Some program terms naturally support decomposing a specification into independent requirements for their subterms.
For example, 
given a goal type $T$ and assuming that the top-level construct of the program is a conditional with a known guard,
we can pass $T$ on to the two branches of the conditional, together with the appropriate path conditions derived from the guard,
and proceed to check (or synthesize) them completely independently.
Unfortunately, for other program terms there might be infinitely many ways to precisely decompose a specification.
Take a function application, $\App{f}{x}$:
the specification $\App{f}{x}\HasT\T{Nat}$
can be satisfied by requiring that $f$ subtract one and $x$ be positive,
or that $f$ add one, and $x$ be greater than negative one, and so on.
The challenge in this case is to find an over-approximate decomposition:
that is, construct requirements on $f$ and $x$ that are necessary but generally not sufficient for the correctness of $f x$,
yet are strong enough to filter out many incorrect subterms.

To address this challenge, we propose a new type checking mechanism for refinement types, 
which we dub \emph{local liquid type checking}.
At the heart of the new mechanism is a type system
inspired by bidirectional type checking~\cite{PierceTu00}.
Bidirectional systems interleave top-down and bottom-up propagation of type information 
depending on the syntactic structure of the program;
in this work we extend the bottom-up phase of bidirectional checking with top-down propagation of over-approximate type information,
resulting in a \emph{round-trip type checking} mechanism, which promotes modular checking of function applications. 
Additionally, we equip the type system with a novel \emph{liquid abduction} rule,
which enables modular checking of branching terms.

Refinement type checking involves
solving subtyping constraints over unknown refinement types.
The modularity requirement precludes our system from using the two techniques employed to this end by liquid type inference---%
Hindley-Milner unification and the least-fixpoint Horn solver---%
since both techniques are designed to work on complete programs and propagate type information bottom-up.
Instead, local liquid type checking incorporates an algorithm for solving subtyping constraints incrementally, as it analyzes different parts of the program.
Most notably, top-down propagation requires finding the \emph{greatest fixpoint} solution to unknown refinements instead of the least,
which is known to be fundamentally more expensive;
we propose a practical implementation for this fixpoint computation, 
which we call \textsc{MUSFix},
inspired by an existing algorithm for enumerating minimal unsatisfiable subsets (MUSes)~\cite{LiffitonPrMaMa15}.

\custompar{Results.}
We have combined local liquid type checking and exhaustive enumeration of program terms
in a prototype program synthesizer called \tool 
(for ``\textsc{Syn}thesis with li\textsc{quid} types''),
which we evaluated on \exCount synthesis problems from a variety of sources.
The implementation, the benchmarks, 
and a web interface for \tool are available from the tool repository~\cite{SynquidRepo}.

Our evaluation indicates that the techniques described above work well
for synthesizing programs that manipulate lists and trees, 
as well as data structures with complex invariants
and custom user-defined data structures.
\tool was able to synthesize programs that are more complex than those previously reported in the literature,
including five different sorting algorithms, and manipulations of binary search trees, AVL trees, and Red-Black trees.
The evaluation also shows that the modularity features of local liquid type checking
and the \textsc{MUSFix} solver are crucial for the performance of the system.

We compare our system with the existing synthesizers 
based on input-output examples~\cite{AlbarghouthiGuKi13,OseraZd15,FeserChDi15,FrankleOWZ16},
and Hoare-style verification~\cite{LeinoMi12,KneussKuKuSu13},
and demonstrate that \tool can handle the majority of its competitors' most challenging benchmarks,
taking a similar or shorter amount of time.
In addition, compared with the example-based tools,
\tool's specifications are usually more concise and the generated solutions are provably correct;
compared with the tools based on Hoare-style reasoning, 
\tool can verify (and thus synthesize) more complex programs thanks to automatic refinement inference.

\section{Overview}\label{sec:overview}

\tool operates within a core ML-like language featuring conditionals, algebraic datatypes, pattern matching, parametric polymorphism, and fixpoints.
We equip the language with \emph{general decidable} refinement types,
closely following the liquid types framework~\cite{RondonKaJh08,KawaguchiRoJh09,VazouRoJh13}.
The type system includes refined base types of the form $\{B\mid\psi\}$,
where $\psi$ is a \emph{refinement predicate} over the program variables and a special \emph{value variable} $\nu$, which does not appear in the program.
Base types can be combined into dependent function types of the form $\funT{x}{T_x}{T_2}$,
where $x$ may appear in the refinement predicates of $T_2$.
Our framework is agnostic to the exact logic of refinement predicates as long as validity of their boolean combinations is decidable;
our prototype implementation uses the quantifier-free logic of arrays\footnote{Arrays are used to model sets.}, uninterpreted functions, and linear integer arithmetic,
which is sufficient for all the examples and benchmarks in this paper.

A \emph{synthesis problem} is defined by
\begin{inparaenum}[(1)] 
\item a goal refinement type $T$
\item a \emph{typing environment} $\env$ and
\item a set of \emph{logical qualifiers} $\quals$.
\end{inparaenum}
A solution to the synthesis problem is a program term $t$ that has the type $T$ in the environment $\env$.
The environment contains type signatures of components available to the synthesizer
(which may include datatype constructors, ``library'' functions, and local variables)
as well as any path conditions that can be assumed when synthesizing $t$.
Qualifiers are predicates from the refinement logic used as building blocks for unknown refinements and branch guards.
Our system extracts an initial set of such predicates automatically from the goal type and the types of components; 
for all our experiments, the automatically extracted qualifiers were sufficient to synthesize all the necessary refinements, 
but in general the user might have to provide additional predicates. 

Given a synthesis problem,
\lang constructs a candidate solution,
by either by either picking a component in $\env$
or decomposing the problem into simpler subproblems and recursively obtaining a solution $t_i$ to each one.
Since the decomposition is generally incomplete,
a candidate obtained by combining $t_i$'s is not guaranteed to have the desired type $T$;
to check if the candidate is indeed a solution,
the system generates a \emph{subtyping constraint}.
If the constraint cannot be satisfied, the system backtracks to pick a different combination of solutions to subproblems
(or a different decomposition altogether);
the stronger the sub-goals produced during decomposition, the less \lang has to backtrack.
The rest of the section illustrates the details of this procedure 
and showcases various features of the specification language on a number of examples.

\custompar{Example 1: Recursive Programs and Condition Abduction.}
We first revisit the \T{replicate} example from the introduction.
We assume that the set of available components includes functions \T{0}, \T{inc} and \T{dec} on integers,
as well as a list datatype whose constructors are refined with length information,
expressed by means of an uninterpreted function (or \emph{measure}) \T{len}.
\begin{nanoml}
0::{Int|_v = 0}
inc::x:Int -> {Int|_v = x + 1}
dec::x:Int -> {Int|_v = x - 1}

termination measure len :: List _b -> Nat
data List _b where 
  Nil :: {List _b | len _v == 0}
  Cons :: _b -> xs: List _b -> 
    {List _b | len _v == len xs + 1}
\end{nanoml}
Measure \T{len} also serves as the \emph{termination metric} on lists (denoted with the \T{termination} keyword above):
it maps lists to a type that has a predefined well-founded order in our language and thus enables termination checks for recursion on lists.
% Examples in this paper integrate measure-based refinements with datatype definitions for simplicity;
% in our tool, following \cite{KawaguchiRoJh09}, measure definitions are syntactically decoupled from their datatypes
% and then desugared automatically into the above representation,
% thus, allowing a single datatype to be decorated with multiple, custom measures depending on the problem at hand.

For the rest of the section, let us fix the set of logical qualifiers $\quals$ to $\{\star\leq \star,\star\neq \star\}$, 
where $\star$ is a placeholder that can be instantiated with any program variable or $\nu$.

Given the specification 
$$
n:\T{Nat} \to x:\alpha \to \{\T{List}\ \alpha\mid \T{len}\ \nu = n\}
$$
\tool picks $\lambda$-abstraction as the top-level construct,
and creates a synthesis subproblem for its body with a simpler goal type $\{\T{List}\ \alpha\mid \T{len}\ \nu = n\}$.
The system need not consider other choices of the top-level construct,
since every terminating program has an equivalent $\beta$-normal $\eta$-long form,
where all functions are fully applied and the head of each application is a variable;
moreover, the above decomposition is precise, since any solution to the subproblem will satisfy the top-level goal.

As part of the decomposition, the arguments $n:\T{Nat}$ and $x:\alpha$ are added to the environment, 
together with the function \T{replicate} itself, to account for the possibility that it may be recursive. 
In order to ensure termination of recursive calls,
the system weakens the type of \T{replicate} in the environment to
$$
n'\colon\{\T{Int}\mid 0 \leq \nu < n\} \to x'\colon\alpha\to \{\T{List}\ \alpha\mid \T{len}\ \nu = n'\}
$$
demanding that the first argument be strictly decreasing.
\lang picks $n$ as the termination metric in this case,
since it is the only argument whose type has an associated well-founded order. 

In the body of the function, 
the top-level construct might be a branching term.
Rather than exploring this possibility explicitly,
\tool adds a fresh \emph{predicate unknown} $P_0$ as a path condition to the environment,
and then searches for a branch-free program term that satisfies the specification assuming $P_0$.
Each candidate branch-free term $t$ is validated by solving a subtyping constraint;
as part of this process, \tool discovers the weakest $P_0$ that makes $t$ satisfy the specification.
In case $t$ is valid unconditionally, the weakest such $P_0$ is \T{True}, and no branch is generated.

Suppose the first branch-free term that \tool considers is \T{Nil};
this choice results in a subtyping constraint
\begin{multline*}
n:\T{Nat}; x:\alpha; P_0 \vdash\\ \{\T{List}\ \beta'\mid \T{len}\ \nu = 0\} \Subt \{\T{List}\ \alpha\mid \T{len}\ \nu = n\}
\end{multline*}
where $\beta'$ is a free type variable.
The constraint imposes two requirements:
\begin{inparaenum}[(i)] 
\item the \emph{shapes} of the two types (i.e. their underlying unrefined types) must have a unifier~\cite{Pierce02} and 
\item the refinements of the subtype must subsume those of the supertype under the assumptions encoded in the environment.
\end{inparaenum}
The constraint above gives rise to a unifier $[\beta'\mapsto\{\alpha\mid P_1\}]$ (where $P_1$ is a fresh predicate unknown)
and two Horn constraints: $P_0 \And P_1 \Implies \top$ and $0 \leq n \And P_0 \And \T{len _v} = 0 \Implies \T{len _v} = n$.
\tool uses the \textsc{MUSFix} Horn solver (\autoref{sec:theory:horn}) to find the weakest assignment of \emph{liquid formulas} to $P_0$ and $P_1$
that satisfies both Horn constraints.
A liquid formula is a conjunction of atomic formulas,
obtained by replacing $\star$-placeholders in each qualifier in $\quals$ with appropriate variables.
If for some $P_i$ no valid assignment exists,
or the weakest valid assignment is a contradiction,
the candidate program is discarded.

In our example, \textsc{MUSFix} discovers the weakest assignment $\liquidass = [P_0 \mapsto n \leq 0, P_1 \mapsto \top]$,
effectively \emph{abducing} the necessary branching condition.
Since the condition is not trivially true,
the system proceeds to synthesize the remaining branch under the path condition $\lnot(n \leq 0)$.
A similar strategy for generating branching programs has been successfully employed in several existing synthesis tools~\cite{LeinoMi12,KneussKuKuSu13,AlbarghouthiGuKi13,AlurCR15}
and is commonly referred to as \emph{condition abduction}.
Each condition abduction technique faces the challenge of searching a large space of potential conditions efficiently;
our approach, which we dub \emph{liquid abduction}, 
addressed this challenge by restricting conditions to liquid formulas
and using \textsc{MUSFix} to explore the space of liquid formulas efficiently.

The remaining branch has to deal with the harder case of $n > 0$.
When enumerating candidates for this branch,
\tool eventually decides to apply the \T{replicate} component (that is, make a recursive call), 
and searches for the parameters via recursive application of the synthesis procedure. 
At this point, the strong precondition on the argument $m$, $0 \leq \nu < n$, which arises from the termination requirement,
enables filtering candidate arguments locally, before synthesizing the rest of the branch.
In particular, the system will discard the candidates $n$ and $\T{inc}\ n$ right away,
since they fail to produce a value strictly less than $n$.

\custompar{Example 2: Complex Data Structures and Invariant Inference.}
Assuming comparison operators in our logic are generic, 
we can define the type of binary search trees as follows:
\begin{nanoml}
termination measure size :: BST _a -> Int
measure keys :: BST _a -> Set _a
data BST _a where
  Empty::{BST _a|keys _v == []}
  Node::x:_a -> l:BST{_a|_v < x} -> r:BST{_a|x < _v} 
    -> {BST _a|keys _v == keys l + keys r + [x]}
\end{nanoml}
According to this definition, one can obtain a \T{BST} either by taking an empty tree,
or by composing a node with key $x$ and two \T{BST}s, $l$ and $r$,
in which all keys are, respectively, strictly less and strictly greater than $x$.
The type is additionally refined by the measure \T{keys}, which denotes the set of all keys in the tree,
and a termination measure \T{size} (size-related refinements are omitted in the interest of space).

The following type specifies insertion into a \T{BST}:
\begin{nanoml}
insert::x:_a -> t:BST _a -> 
                  {BST _a|keys _v == keys t + [x]}
\end{nanoml}
From this specification, \tool generates the following implementation within two seconds:
\begin{nanoml}
insert = \x.\t.match t with
  | Empty -> Node x Empty Empty
  | Node y l r -> if x <= y && y <= x
      then t
      else if y <= x          
        then Node y l (insert x r)
        else Node y (insert x l) r
\end{nanoml}

Pattern matching in this example is synthesized using a special case of liquid abduction:
type-checking the term \T{Node x Empty Empty} against the goal type $\{\T{BST}\ \alpha\mid\T{keys}\ \nu = \T{keys}\ t + [x]\}$,
causes the system to abduce the condition $\T{keys}\ t = []$,
which implies a match on $t$.

The challenging aspect of this example is reasoning about sortedness.
For example, for the term $\T{Node}\ y\ l\ (\T{insert}\ x\ r)$ to be type-correct,
the recursive call must return the type $\T{BST}\ \{\alpha\mid y < \nu\}$.
This type does not appear explicitly in the user-provided signature for \T{insert};
in fact, verifying this program requires discovering a nontrivial inductive invariant of \T{insert}
(that adding a key greater than some value $z$ into a tree with keys greater than $z$ again produces a tree with keys greater than $z$),
which puts this and similar examples beyond reach of existing synthesizers based on Hoare-style reasoning~\cite{LeinoMi12,KneussKuKuSu13}.

In our framework, this property is easily inferred by the Horn constraint solver in combination with polymorphic recursion.
When \T{insert} is added to the environment, 
its type is generalized to $\forall \beta . x\colon\beta \to u\colon \{\T{BST}\ \beta\mid \T{size}\ u < \T{size}\ t \} \to \{\T{BST}\ \beta\mid\T{keys}\ \nu = \T{keys}\ t + \{x\}\}$.
At the site of the recursive call, the precondition of $\T{Node}\ y\ l$ imposes a constraint that simplifies to $\T{BST}\ \beta\Subt \T{BST}\ \{\alpha\mid\ y < \nu\}$,
which leads to instantiating $[\beta\mapsto\{\alpha\mid P_0\}]$ and $[P_0\mapsto y < \nu]$.

Importantly, due to \emph{round-trip type checking} (\autoref{sec:theory:rules}), 
this assignment is discovered before the two arguments to $\T{insert}$ are synthesized,
which has the effect of propagating the requirement imposed by $\T{Node}$ \emph{top-down} through the application of $\T{insert}$ onto its arguments.
In particular, using the goal type $\{\alpha\mid\ y < \nu\}$ for the first argument of \T{insert},
the system can immediately discard the candidate \T{y},
while trying \T{x} succeeds and leads to the abduction of the branch condition $y \leq x$.
As our evaluation shows, disabling this type of early filtering increases the synthesis time for this example from less than two seconds to over two minutes.

\custompar{Example 3: Abstract Refinements.}
Using refinement types as an interface to synthesis raises the question of their expressiveness.
Restricting refinements to decidable logics fundamentally limits the class of programs they can fully specify,
and for other programs writing a refinement type might be possible but cumbersome compared to providing a set of input-output examples or a specification in a richer language.
The previous examples suggest that refinement types are effective for specifying programs that manipulate data structures
with nontrivial universal and inductive invariants.
In this example we demonstrate how extending the type system with \emph{abstract refinements} allows us to express a wider class of properties,
for example, talk about the order of list elements.

Abstract refinements, proposed in~\cite{VazouRoJh13}, enable explicit quantification over refinements of datatypes and function types.
For example, a list datatype can be parameterized by a binary relation $r$
that must hold between any in-order pair of elements in the list:
\begin{nanoml}
data RList _a <|r::_a -> _a -> Bool|> where
 Nil::RList _a <|r|>
 Cons::x:_a -> RList {_a|r x _v} <|r|> -> RList _a <|r|>
\end{nanoml}

On the one hand this enables concise definitions of lists with various inductive properties as instantiations of \T{RList}:
\begin{nanoml}
IList _a = RList _a <|\x\y.x <= y|> -- Increasing list
UList _a = RList _a <|\x\y.x != y|> -- Unique list
List _a = RList _a <|\x\y.True|> -- Unrestricted list
\end{nanoml}

On the other hand, making list-manipulating functions polymorphic in this relation, provides an elegant way to specify order-related properties.
Consider the following type for list reversal: 
\begin{nanoml}
reverse::<|r::_a -> _a -> Bool|>.xs:RList _a <|r|> -> 
  {RList _a <|\x\y.r y x|>|len _v == len xs}
\end{nanoml}
It says that whatever relation holds between every \emph{in-order} pair of elements of the input list,
also has to hold between every \emph{out-of-order} pair of elements of the output list.
This type does not restrict the applicability of \T{reverse},
since at the call site $r$ can always be instantiated with \T{True};
the implementation of \T{reverse}, however, has to be correct to any value of $r$,
which leaves the synthesizer no choice but reverse the order of list elements.
Given the above specification and a component that appends an element to the end of the list (specified in a similar fashion), 
\tool synthesizes the standard implementation of list reversal.

\custompar{Example 4: Higher-Order Combinators and Auxiliary Function Discovery.}
Complex programs might require recursive auxiliary functions.
Discovering specifications for such functions automatically is a difficult task,
akin to \emph{lemma discovery} in theorem proving~\cite{Montano-RivasMDB12,ClaessenJRS12,HerasKJM13},
which largely remains an open problem. % but in practice rarely can be done without user interaction.
\tool expects users to provide the high-level insight about a complex algorithm in the form of auxiliary function signatures.
For example, if the goal is to synthesize a list sorting function with the following signature
\begin{nanoml}
sort::xs:List _a -> {IList _a|elems _v == elems xs}
\end{nanoml}
(where \T{elems} denotes the set of list elements),
the user can express the insight behind different sorting algorithms by providing different auxiliary functions:
insertion into a sorted list for \emph{insertion sort},
splitting and merging for \emph{merge sort},
or partitioning and concatenation for \emph{quick sort}.
Naturally, the implementation of the auxiliary functions can in turn be synthesized,
but coming up with their specification is the creative step that generally requires user interaction,
and can be considered a major hurdle on the path to fully automatic synthesis.

It turns out, however, that replacing general recursion with higher-order combinators such as \T{map} and \T{fold}%
---a style widely used and highly encouraged in functional programming---%
makes it possible to infer requirements on the auxiliary function from the specification of the main program.
This is one of the main insights behind the synthesizer \textsc{$\lambda^2$}~\cite{FeserChDi15},
which relies on hard-coded rules for propagating input-output examples top-down through common combinators.
\tool supports this top-down propagation of specifications out of the box 
thanks to the combination of refinement types and polymorphism. 

Consider the following type for function \T{foldr}, 
which folds a binary operation \T{f} over a list \T{ys} from right to left:
\begin{nanoml}
foldr::<p::List _b -> _c -> Bool>.
  f:(t:List _b -> h:_b -> acc:{_c|p t _v} -> 
        {_c|p (Cons h t) _v}) ->
  seed:{_c|p Nil _v} ->
  ys:List _b -> {_c|p ys _v}
\end{nanoml}
The shape of this type is slightly different from the usual signature of \T{foldr}:
the operation \T{f} takes an extra ghost argument \T{t}, 
which
denotes the part of the list that has already been folded%
\footnote{Extending \tool with \emph{bounded refinement types}~\cite{VazouBJ15} would enable a more natural specification without the ghost argument.}.
The type of is parametrized by a binary relation \T{p} that \T{foldr} establishes between the input list \T{ys} and the output;
it requires that the relationship hold between the empty list and the \T{seed},
and that applying \T{f} to a head element \T{h} and the result of folding a tail list \T{t} yield a result that satisfies the relationship with \T{Cons h t}
(in other words, \T{p} plays the role of a loop invariant).
Note that folding a list left-to-right (\T{foldl}) requires a more complex specification
that cannot currently be expressed within the \tool type system.

What happens if we ask \tool to synthesize \T{sort}, while providing \T{foldr} as the only component?
When trying out an application of \T{foldr}, 
round-trip type checking handles its higher-order argument, \T{f}, in a special way,
since in our type system, as in~\cite{RondonKaJh08}, 
\T{f} cannot appear in the result type of \T{foldr}.
Consequently, the exact value of \T{f} is not required to determine the type of the application,
which gives \tool the freedom to synthesize it independently from the rest of the program.

The tool quickly figures out that \T{foldr ?? Nil xs} has the required type \T{\{IList _a|elems _v == elems xs\}},
given the following assignment to \T{foldr}'s type and predicate variables: 
$[\beta\mapsto \alpha, \gamma\mapsto \T{IList}\ \alpha, \T{p}\mapsto \lambda as . \lambda bs. \T{elems}\ bs = \T{elems}\ as]$.
Now that \tool comes back to the task of filling in the first argument of \T{foldr},
its required type has been determined entirely as
\begin{nanoml}
t:List _a -> h:_a -> 
  acc:{IList _a|elems _v == elems t} -> 
  {IList _a|elems _v == elems (Cons h t)}
\end{nanoml}
(where \T{elems (Cons h t)} is expanded into \T{[h] + elems t} using the definition of the \T{elems} measure in the type of \T{Cons});
in other words, the auxiliary function must insert \T{h} into a sorted list \T{acc}.
Treating this inferred signature as an independent synthesis goal, 
\tool easily synthesizes a recursive program for insertion into a sorted list,
and thus completes the following implementation of insertion sort without requiring any hints from the user,
apart from a general recursion scheme:
\begin{nanoml}
sort = \xs.foldr f Nil xs
  where f = \t.\h.\acc. 
    match acc with
      Nil -> Cons h Nil
      Cons z zs -> if h <= z
        then Cons h (Cons z zs)
        else Cons z (f zs h zs)  
\end{nanoml}       

The next section gives a formal account of the \lang language and type system,
and details its modular type checking mechanism, which enables scalable synthesis.

\section{The \lang Language}\label{sec:theory}

The central goal of this section is to develop a type checking algorithm for a core programming language with refinement types
that is geared towards candidate validation in the context of synthesis.
This context imposes two important requirements on the type checking mechanism
which are necessary for the synthesis procedure to be automatic and scalable.
The first one has to do with the amount of type inference:
the mechanism can expect top-level type annotations---this is how users specify synthesis goals---%
but cannot rely on any annotations beyond that; 
in particular, the types of all polymorphic instantiations and arguments of anonymous functions must be inferred.
The second requirement is to detect type errors \emph{locally}:
intuitively, if a subterm of a program causes a type error independently of its context,
the algorithm should be able to report that error without analyzing the context.

We build our type checking mechanism as an extension to the the liquid types framework~\cite{RondonKaJh08,KawaguchiRoJh09}, 
which uses a combination of Hidley-Milner unification and a Horn solver to infer refinement types.
The original liquid type inference algorithm is not designed for synthesis, 
and thus makes different trade-offs:
in particular it does not satisfy the locality requirement.
Our type checking mechanism achieves locality based on three key ideas.
First, we apply bidirectional type checking~\cite{PierceTu00} to refinement types
and reinforce it with additional top-down propagation of type information,
arriving at \emph{round-trip type checking} (\autoref{sec:theory:rules});
we then further improve locality of rules for for function applications and branching statements (\autoref{sec:theory:extensions}).
Second, we develop a new algorithm for converting subtyping constraints into horn clauses, 
which is able to do so incrementally as the constraints are issued before analyzing the whole program (\autoref{sec:theory:subtyping}).
Finally, we propose a new, efficient implementation for a greatest-fixpoint Horn solver (\autoref{sec:theory:horn}).
In the interest of space we omit abstract refinements (see~\autoref{sec:overview}) from the formalization;
\cite{VazouRoJh13} has shown that integrating this mechanism into the type system that already supports parametric polymorphism is straightforward.

In \autoref{sec:theory:synthesis} we derive \emph{synthesis rules} from the modular type checking rules;
in doing so we follow previous work on type-directed synthesis~\cite{OseraZd15,InalaQLS15},
which has shown how to turn type checking rules for a language into synthesis rules for the same language.

\subsection{Syntax and Types}\label{sec:theory:syntax}

\begin{figure}
\small
$$
\begin{array}{llll}
% Formulas
\psi &::=                                                                                   & \text{\emph{Refinement term:}}\\
  &\quad   \mid \top \mid \bot \mid 0 \mid + \mid \ldots\ \text{\emph{(varies)}}            & \text{interpreted symbol}\\
  &\quad   \mid x                                                                           & \text{uninterpreted symbol}\\
  &\quad   \mid \App{\psi}{\psi}                                                            & \text{application}\\
\\
% Sorts
\Delta &::=                                                                                 & \text{\emph{Sort:}}\\
  &\quad   \mid \mathbb{B} \mid \mathbb{Z}\mid \ldots\ \text{\emph{(varies)}}               & \text{interpreted}\\ 
  &\quad   \mid \delta                                                                      & \text{uninterpreted}\\
% \\    
% Terms    
t &::= e \mid b \mid f                                                                      & \text{\emph{Program term}}\\  
% \\
e &::=                                                                                      & \text{\emph{E-term:}}\\
  &\quad   \mid x                                                                           & \text{variable}\\
  &\quad   \mid \App{e}{e} \mid \App{e}{f}                                                  & \text{application}\\
b &::=                                                                                      & \text{\emph{Branching term:}}\\
  &\quad   \mid \T{if}\ e\ \T{then}\ t\ \T{else}\ t                                         & \text{conditional}\\
  &\quad   \mid \T{match}\ e\ \T{with}\ |_i\ \T{C}_i \langle x_i^j\rangle \mapsto t_i       & \text{match}\\
f &::=                                                                                      & \text{\emph{Function term:}}\\
  &\quad   \mid \lambda x.t                                                                 & \text{abstraction}\\
  &\quad   \mid \T{fix}\ x.t                                                                & \text{fixpoint}\\  
% Types
B &::=                                                                                      & \text{\emph{Base type:}}\\  
  &\quad   \mid \T{Bool}\mid \T{Int}                                                        & \text{primitive}\\   
  &\quad   \mid D\ T_i                                                                      & \text{datatype}\\  
  &\quad   \mid \alpha                                                                      & \text{type variable}\\
% \\  
T &::=                                                                                      & \text{\emph{Type:}}\\
  &\quad  \mid \{B\ \mid\psi\}                                                              & \text{scalar}\\
  &\quad  \mid \funT{x}{T}{T}                                                               & \text{function}\\
S &::=  \forall\alpha_i.T                                                                   & \text{\emph{Type schema}}\\
C &::=  \cdot \mid x\colon T;C                                                              & \text{\emph{Context}}\\
\hat{T} &::=  \contT{C}{T}                                                                  & \text{\emph{Contextual Type}}
\end{array}
$$
\caption{Terms and types.}\label{fig:syntax}
\end{figure}

\autoref{fig:syntax} shows the syntax of the \lang language.

\custompar{Terms.}
Unlike previous work, we differentiate between the languages of refinements and programs.
The former consists of refinement terms $\psi$, which have sorts $\Delta$;
the exact set of interpreted symbols and sorts depends on the chosen refinement logic.
We refer to refinement terms of the Boolean sort $\mathbb{B}$ as formulas.

The language of programs consists of program terms $t$, 
which we split, following~\cite{OseraZd15} into \emph{elimination} and \emph{introduction} terms (E-terms and I-terms for short).
Intuitively, E-terms---variables and applications---propagate type information bottom-up,
\emph{composing} a complex property from properties of their components;
I-terms propagate type information top-down,
\emph{decomposing} a complex requirement into simpler requirements for their components.
Note that conditional guards, match scrutinees, and left-hand sides of applications are restricted to E-terms.
We further separate I-terms into branching terms---conditionals and matches---and function terms---abstractions and fixpoints---%
and disallow branching terms on the right-hand side of application.
This normal form is required to enable precise and efficient local type checking, as explained below.
It does not fundamentally restrict the expressiveness of the language:
every terminating program in lambda calculus can be translated to \lang
by first applying a standard $\beta$-normal $\eta$-long form~\cite{GveroKuKuPi13,OseraZd15}
and then pushing branching terms outside of applications, guards, and scrutinees.

\custompar{Types and Schemas.}
A \lang type is either a scalar---base type refined with a formula---or a dependent function type.
Base types include primitives, type variables, and user-defined datatypes with zero or more type parameters. 
Datatype constructors are represented simply as functions
that must have the type $\forall\alpha_1\ldots\alpha_m.T_1\to\ldots\to T_k\to D\ \alpha_1\ldots\alpha_m$.
A \emph{contextual type} is a pair of a sequence of variable bindings and a type that can mention those variables;
contextual types are useful for precise type checking of applications, as explained in \autoref{sec:theory:rules}.

\lang features ML-style polymorphism, where type variables are universally quantified at the outermost level to yield type schemas.
Unlike ML, we restrict type variables to range only over scalars,
which gives us the ability to determine whether a type is a scalar, even if it contains free type variables.
We found this restriction not to be too limiting in practice.

\begin{figure}
\small
\textbf{Well-Formed Types}\quad$\boxed{\env \vdash \hat{T}}$
\begin{gather*}
\inference[\textsc{WF-Sc}]
{\env;\nu:B \vdash \psi}
{\env \vdash \{B\mid \psi\} }
\quad
\inference[\textsc{WF-Ctx}]
{\env;C\vdash T}
{\env \vdash \contT{C}{T}}
\\
\inference[\textsc{WF-FO}]
{\env\vdash\{B\mid \psi\} & \env;x:\{B\mid \psi\} \vdash T}
{\env \vdash \funT{x}{\{B\mid \psi\}}{T} }
\\
\inference[\textsc{WF-HO}]
{T_x\ \text{non-scalar} & \env \vdash T_x & \env \vdash T}
{\env \vdash T_x \to T }
\end{gather*}
\textbf{Subtyping}\quad$\boxed{\env \vdash T\Subt T'}$
\begin{gather*}
\inference[\textsc{$\Subt$-Sc}]
{\env \vdash B\;\Subt\; B' & \mathsf{Valid}(\llbracket\env\rrbracket_{\psi\Implies\psi'}\And\psi\Implies\psi')}
{\env \vdash \{B\mid \psi\} \;\Subt\; \{B'\mid \psi'\} }
\\
\inference[\textsc{$\Subt$-Fun}]
{\env\vdash T_y\;\Subt\; T_x & \env;y: T_y \vdash [y/x]T\;\Subt\; T'}
{\env \vdash \funT{x}{T_x}{T} \;\Subt\; \funT{y}{T_y}{T'} }
\\
\inference[\textsc{$\Subt$-DT}]
{\env\vdash T_i \;\Subt\; T'_i}
{\env \vdash D\ T_i \;\Subt\; D\ T'_i}
\quad
\inference[\textsc{$\Subt$-Refl}]
{}
{\env \vdash B \;\Subt\; B}
\end{gather*}
\caption{Well-formedness and subtyping.}\label{fig:wf-subt}
\end{figure}

\custompar{Environments, Well-Formedness, and Subtyping.}
A \emph{typing environment} $\env$ is a sequence of
variable bindings $x: T$ and path conditions $\psi$;
we denote conjunction of all path conditions in an environment as $\pathc{\env}$.
A formula $\psi$ is \emph{well-formed} in the environment $\env$,
written $\env \vdash \psi$,
if it is of a Boolean sort and each of its free variables is bound in $\env$ to a type that is consistent with its sort in $\psi$.  
Well-formedness extends to types as shown in \autoref{fig:wf-subt}.
Note the two different rules for first-order and higher-order function types:
in a function type $\funT{x}{T_1}{T_2}$, $T_2$ may reference the formal argument $x$ only if $T_1$ is a scalar type
(that is, only first-order function types are dependent).

The \emph{subtyping} relation $\env \vdash T \Subt T'$ is relatively standard (\autoref{fig:wf-subt}).
For simplicity, we consider all datatypes covariant in their type parameters (rule \textsc{$\Subt$-DT});
if need be, variance can be selected per type parameter 
depending on whether it appears positively or negatively in the constructors.
The crucial part is the rule \textsc{$\Subt$-Sc}, which reduces subtyping between scalar types to implication between their refinements,
under the assumptions extracted from the environment.
Since the refinements are drawn from a decidable logic, this implication is decidable.
The function that extracts assumptions from the environment is parametrized by a formula $\psi$
and returns a conjunction of all path conditions and refinements of all variables mentioned in $\psi$ or the path conditions:
$$
\llbracket\Gamma\rrbracket_{\psi} = \pathc{\env} \And \bnd{\fv(\pathc{\env}) \cup \fv(\psi)}{\env}
$$
where
\begin{align*}
\bnd{v}{\env;x:\{B\mid\psi\}} &= 
  \begin{cases}
    [x/\nu]\psi \And \bnd{v\setminus \{x\} \cup \fv(\psi)}{\env} \quad(x\in v)\\
    \bnd{v}{\env} \quad\text{(otherwise)}
  \end{cases}\\
\bnd{v}{\env;x: T} &= \bnd{v}{\env} \quad\text{($T$ non-scalar)}\\
\bnd{v}{\cdot} &= \top
\end{align*}
This definition limits the effect of an environment variable with an inconsistent refinement
to only those subtyping judgments that (transitively) mention that variable.

\subsection{Round-Trip Type Checking}\label{sec:theory:rules}

This section describes the core of \lang's type system.
It is inspired by bidirectional type checking~\cite{PierceTu00},
which interleaves top-down and bottom-up propagation of type information 
depending on the syntactic structure of the program,
with the goal of making type checks more local.
Bidirectional typing rules use two kinds of typing judgments:
an \emph{inference} judgment, written $\env \vdash e \produce T$,
states that the term $t$ generates type $T$ in the environment $\env$;
a \emph{checking} judgment, $\env \vdash t \consume T$,
states that the term $t$ checks against a known type $T$ in the environment $\env$.
Accordingly, all typing rules can be split into inference and checking rules, 
depending on the judgment they derive.
\iflong
Bidirectional type checking rules for \lang are given in \autoref{app:proofs}.
\else
Bidirectional type checking rules for \lang can be found in the technical report~\cite{Techreport}.
\fi

In a bidirectional system, analyzing a program starts with propagating its top-level type annotation top-down using checking rules,
until the system encounters a term $t$ to which no checking rule applies.
At this point the system switches to bottom-up mode, infers the type $T'$ of $t$,
and checks if $T'$ is a subtype of the goal type;
if the check fails, $t$ is rejected.
Bidirectional type propagation is ``all-or-nothing'':
once a checking problem for a term cannot be decomposed perfectly into checking problems for its subterms,
the system abandons all information about the goal type and switches to purely bottom-up inference.
Our insight is that some information from the goal type can be retained in the bottom-up phase,
leading to more local error detection.
To this end, we modify the bidirectional inference judgment into a \emph{strengthening} judgment $\env \vdash t \consume T \produce T'$,
which reads as follows: 
in the environment $\env$, term $t$ checks against a known type $T$ \emph{and} generates a stronger type $T'$.
We call the resulting type system \emph{round-trip},
since it propagates types top-down and then back up.

\begin{figure}
\small
\textbf{Type Strengthening}\quad$\boxed{\env \entailsQ e \consume T \produce \hat{T'}}$
\begin{gather*}
\inference[\textsc{VarSc}]
{\env(x) =\{B \mid \psi\} & \env\vdash\{B \mid \psi\}\Subt T}
{\env \entailsQ x \consume T \produce \{B\mid\nu = x\} }
\\
\inference[\textsc{Var$\forall$}]
{\env(x)=\forall\alpha_i.T'\ &  \env \entailsQ T_i & \env\vdash[T_i/\alpha]T' \Subt T}
{\env \entailsQ x \consume T \produce [T_i/\alpha]T'}
\\
\inference[\textsc{AppFO}]
{\env \entailsQ e_1 \consume \{B \mid \bot\}\to T \produce \contT{C_1}{(\funT{x}{\{B \mid \psi\}}{T'})}    \\  
\env;C_1 \entailsQ e_2 \consume \{B \mid \psi\} \produce \contT{C_2}{T_x} \\
\env;C_1;C_2;x\colon T_x \vdash T' \Subt T}
{\env \entailsQ \App{e_1}{e_2} \consume T \produce \contT{C_1;C_2;x\colon T_x}{T'} }
\\
\inference[\textsc{AppHO}]
{\env \entailsQ e \consume \tbot\to T \produce \contT{C}{(T'_x \to T')}    \\  \env;C \entailsQ f \consume T'_x}
{\env \entailsQ \App{e}{f} \consume T \produce \contT{C}{T'}}
\end{gather*}

\textbf{Type Checking}\quad$\boxed{\env \entailsQ t \consume S}$
\begin{gather*}
\inference[\textsc{IE}]
{\env \entailsQ e \consume T \produce \hat{T'}}
{\env \entailsQ e \consume T}
\\
\inference[\textsc{Abs}]
{\env;y\colon T_x \entailsQ t \consume [y/x]T}
{\env \entailsQ \lambda y.t \consume (\funT{x}{T_x}{T})}
\\
\inference[\textsc{If}]
{\env \entailsQ e \consume \T{Bool} \produce \contT{C}{\{\T{Bool}\mid \psi\}} \\
\env;C;[\top/\nu]\psi \entailsQ t_1 \consume T  &  
\env;C;[\bot/\nu]\psi \entailsQ t_2 \consume T}
{\env \entailsQ\T{if}\ e\ \T{then}\ t_1\ \T{else}\ t_2 \consume T}
\\
\inference[\textsc{Match}]
{\env \entailsQ e \consume \ttop \produce \contT{C}{\{D\ T_k\mid\psi\}}  \\
\T{C}_i = T_i^j \to \{D\ T_k\mid\psi'_i\}    &
\env_i = \{x_i^j\colon T_i^j\};[x'/\nu]\psi'_i \\
\env;C;[x'/\nu]\psi;\env_i \entailsQ t_i \consume T}
{\env \entailsQ \T{match}\ e\ \T{with}\ |_i\ \T{C}_i \langle x_i^j\rangle \mapsto t_i \consume T}
\\
\inference[\textsc{TAbs}]
{\env \entailsQ t \consume T  &  \alpha_i\ \text{not free in}\ \env}
{\env \entailsQ t \consume \forall\alpha_i.T}
\\
\inference[\textsc{Fix}]
{\env;x\colon S^\before \entailsQ t \consume S}
{\env \entailsQ \T{fix}\ x.t \consume S}
\end{gather*}
\caption{Rules of round-trip type checking.}\label{fig:round-trip}
\end{figure}

Derivation rules for round-trip type checking are presented in \autoref{fig:round-trip}.
All judgments are parametrized by the set of qualifiers $\quals$,
used to construct unknown refinements as explained below.
Checking rules encode the way a checking judgment for an I-term $t$ is decomposed into simpler checking judgments for its components.
Strengthening rules encode the way a goal type for an E-term $e$ is decomposed into \emph{over-approximate} goal types for its subterms, 
which are necessary but in general not sufficient for correctness,
while the precise type of $e$ is constructed from the inferred types of its subterms.
A round-trip type checker starts with a top-down phase, just as a bidirectional one would;
when it encounters an E-term, it applies the corresponding strengthening rule and discards the inferred type (see rule \textsc{IE}).
Thus, instead of detecting type errors at the boundary between I- and E-terms,
the round-trip system performs local checks for each variable and function application.

In order to support goals types with an underspecified shape
(as required for match scrutinees and higher-order applications),
we augment \lang with $\ttop$ and $\tbot$ types,
which are, respectively, a supertype and a subtype of every type.
Note that these types are ignored when computing the logical representation of the environment $\llbracket\env\rrbracket_\psi$,
since they are not considered scalar.
Also note that the precision of round-trip type checking crucially relies
on the fact that only E-terms appear in strengthening judgments;
this is why \lang bans branching terms from function arguments, conditional guards, and match scrutinees.

\custompar{Polymorphic instantiations.}
The rule \textsc{Var$\forall$}, which handles polymorphic instantiations,
replaces type variables $\alpha_i$ with types $T_i$ chosen nondeterministically to satisfy all subtyping checks%
\footnote{The same rule handles monomorphic non-scalar variables, assuming zero type variables.}.
In order to tame this nondeterminism, following~\cite{RondonKaJh08}, we restrict $T_i$s to \emph{liquid types}.
A formula $\psi$ is \emph{liquid} in $\env$ with qualifiers $\quals$,
written $\env \entailsQ \psi$,
if it is a conjunction of well-formed formulas, 
each of which is obtained from a qualifier in $\quals$ by substituting $\star$-placeholders with variables.
This notion extends to types, $\env \entailsQ T$, in a way analogous to well-formedness (\autoref{fig:wf-subt}).
Note that the set of all liquid formulas in a given environment is finite,
and so is the set of all liquid types with a fixed shape.
\autoref{sec:theory:subtyping} and \autoref{sec:theory:horn} present a deterministic algorithm for finding the types $T_i$. 

\custompar{Applications.}
The application rules \textsc{AppFO} and \text{AppHO} are the core of the round-trip type system:
they are responsible for propagating partial type information down to the left-hand side of an application.
The type system distinguishes between first-order and higher-order applications,
since in a function type $\funT{x}{T_1}{T_2}$, $T_2$ cannot mention $x$ if $T_1$ is a function type (see \autoref{fig:wf-subt}).
As a result, a higher-order application always yields the type $T_2$ independently of the argument.
If instead $T_1$ is a scalar type, we have to replace $x$ inside $T_2$ with the actual argument of the application.
Unfortunately, we cannot assign the application $\App{e_1}{e_2}$ the type $[e_2/x]T_2$,
since $e_2$ is a program term, which does not necessarily have a corresponding refinement precisely capturing its semantics.
We address this problem by assigning $\App{e_1}{e_2}$ a \emph{contextual type} $\contT{C}{T_2}$,
where the context $C$ binds the variable $x$ to the precise type of $e_2$.

\begin{example*}
We demonstrate the local error detection enabled by rule \textsc{AppFO}
on the following type-checking problem:
$$
\env\entailsQ \App{\App{\T{append}}{xs}}{xs} \consume \{\T{List Pos}\mid \App{\T{len}}{\nu} = 5\}
$$
where \T{Pos} is an abbreviation for $\{\T{Int}\mid\nu > 0\}$
and $\env$ contains the following bindings:
\begin{align*}
&xs: \{\T{List Nat}\mid \App{\T{len}}{\nu} = 2\};\\
&\T{append}: \forall\alpha . 
              \begin{aligned}[t]
              &l : \{\T{List}\  \alpha\mid \App{\T{len}}{\nu} \geq 0\} \\
              &\to r : \{\T{List}\  \alpha\mid \App{\T{len}}{\nu} \geq 0\}\\ 
              &\to \{\T{List}\  \alpha\mid \App{\T{len}}{\nu}=\App{\T{len}}{l} + \App{\T{len}}{r}\}
             \end{aligned}
\end{align*}
Intuitively, the constraint on the length of the output list is hard to verify without analyzing the whole expression,
while the mismatch in the type of the list elements can be easily found without considering the second argument of \T{append}.
Refinement types provide precise means to distinguish those cases:
the length-related refinement of \T{append} is dependent on the arguments $l$ and $r$,
whereas the type of the list elements cannot possibly mention $l$ or $r$,
since it has to be well-formed in a scope where these variables are not defined.

Applying the \textsc{AppFO} rule twice to the judgment above yields
$\env\entailsQ\T{append}\consume \{B_0\mid\bot\} \to \{B_1\mid\bot\} \to \{\T{List Pos}\mid \T{len}\ \nu = 2\}$,
where the base types $B_0$ and $B_1$ are yet to be inferred.
Applying \textsc{Var$\forall$}, and decomposing the resulting subtyping check with \textsc{$\Subt$-Fun}, we get
\begin{align*}
&\env;l: \{B_0\mid\bot\};r: \{B_1\mid\bot\} \vdash\\
&\{\T{List}\  T_0\mid \App{\T{len}}{\nu}=\App{\T{len}}{l} + \App{\T{len}}{r}\} \Subt \{\T{List Pos}\mid \T{len}\ \nu = 2\}
\end{align*}
Using \textsc{$\Subt$-Sc}, this judgment can be decomposed into an implication on refinements---vacuous thanks to the types of $l$ and $r$---%
and subtyping on base types, $\T{List}\  T_0 \Subt \T{List Pos}$,
which is not vacuous since here $l$ and $r$ are out of scope.
The first argument of \T{append} is checked against the type $\T{List}\  T_0$ (in the second premise of \textsc{AppFO}),
which imposes a subtyping check $\env\vdash \T{List Nat}\  \Subt \T{List}\  T_0$.
Since no type $T_0$ satisfies both subtyping relations,
the type checker rejects the term $\T{append}\ xs$.
\end{example*}

\custompar{Recursion.}
Another rule in \autoref{fig:round-trip} that deserves some discussion is \textsc{Fix}, which comes with a termination check.
In the context of synthesis, termination concerns are impossible to ignore,
since non-terminating recursive programs are always simpler than terminating ones, 
and thus would be synthesized first if considered correct.  
The \textsc{Fix} rule gives the ``recursive call'' a termination-weakened type $S^\before$,
which intuitively denotes ``$S$ with strictly smaller arguments''.
The exact definition of termination-weakening is a parameter to our system.
Our implementation provides a predefined well-founded order on primitive base types,
and allows the user to define one on datatypes by mapping them to primitive types using termination metrics;
then $S^\before$ is defined as a lexicographic order on the tuple of all arguments of $S$ that have an associated well-founded order.

\subsection{Soundness and Completeness}\label{sec:theory:proofs}

We show soundness and completeness of round-trip type checking
\emph{relative} to purely bottom-up liquid type inference~\cite{RondonKaJh08}.
\iflong
For detailed proofs see \autoref{app:proofs}.
\else
Detailed proofs are available in the technical report~\cite{Techreport}.
\fi

Round-trip type checking is \emph{sound} in the sense that 
whenever a \lang term $t$ type-checks against a schema $S$, $\env\entailsQ t\consume S$,
there exist a set of qualifiers $\quals'$ and a schema $S'$,
such that the bottom-up system infers $S'$ for $t$, $\env\vdash_{\quals'} t :: S$,
and $\env\vdash S' \Subt S$.
Note that bottom-up inference might require strictly more qualifiers than type checking: 
in the bottom-up system, generating types for branching statements and abstractions
imposes the requirement that these types be liquid;
the round-trip system obtains the types of those terms by decomposing the goal type,
thus the liquid restriction does not apply.
In practice the difference is irrelevant, 
since the type inference algorithm can extract the missing qualifiers 
from the top-level goal type and the preconditions of component functions.
Thus, if $\quals$ contains a sufficient set of qualifiers
such that the goals schema is liquid ($\env\entailsQ S$) 
and the preconditions of component function are liquid, which we denote as $\entailsQ\env$,
then we can take $\quals' = \quals$.

\begin{theorem}[Soundness of round-trip type checking]
If $\entailsQ \env$, $\env\entailsQ S$,
and $\env \entailsQ t \consume S$, then $\env\entailsQ t :: S'$ and $\env \entailsQ S'\Subt S$.
\end{theorem}

Unlike liquid type inference, the round-trip system requires a proof of termination for all fixpoints;
thus if $\env\entailsQ t :: S$,
but $t$'s termination cannot be shown using the chosen definition of termination weakening,
the round-trip type system will reject $t$.
Thus we show completeness for a \emph{weakened} round-trip system,
obtained from \autoref{fig:round-trip} by replacing $S^{\before}$ in the premise of the \textsc{Fix} rule by $S$.
We denote the checking judgment of the modified system as $\env \entailsQ^* t \consume S$.

\begin{theorem}[Completeness of round-trip type checking]
If $\env\entailsQ t :: S$, then $\env \entailsQ^* t \consume S$.
\end{theorem}

\subsection{Type System Extensions}\label{sec:theory:extensions}

In this section we further improve the locality of type checking for
function applications and branching terms.

\custompar{Type Consistency.}
Recall the type checking problem 
$$
\env\entailsQ \App{\App{\T{append}}{xs}}{xs} \consume \{\T{List Pos}\mid \App{\T{len}}{\nu} = 5\}
$$
from \autoref{sec:theory:rules}, %, where we could not reason about the length of $\App{\App{\T{append}}{xs}}{xs}$ before analyzing both arguments of the call.
and let us change the type of $xs$ to $\{\T{List Pos}\mid \App{\T{len}}{\nu} = 6\}$.
In this case, $xs$ has the right element type, \T{Pos},
but intuitively the partial application $\App{\T{append}}{xs}$ can still be safely rejected,
since no second argument with a non-negative length can fulfill the goal type.

\begin{figure}
\small
\textbf{Consistency}\quad$\boxed{\env \vdash T\Consi T'}$
\begin{gather*}
\inference[\textsc{$\Consi$-Sc}]
{\env \vdash B\Consi B' & \mathsf{Sat}(\llbracket\Gamma\rrbracket_{\psi \And \psi'}\And\psi\And\psi')}
{\env \vdash \{B\mid \psi\} \Consi \{B'\mid \psi'\} }
\\
\inference[\textsc{$\Consi$-Fun}]
{\env;x: T_x \vdash T\Consi [x/y]T'}
{\env \vdash \funT{x}{T_x}{T} \Consi \funT{y}{T_y}{T'} }
\\
\inference[\textsc{$\Consi$-DT}]
{\env\vdash T_i \Consi T'_i}
{\env \vdash D\ T_i \Consi D\ T'_i}
\quad
\inference[\textsc{$\Consi$-Refl}]
{}
{\env \vdash B \Consi B}
\end{gather*}
\caption{Type consistency.}\label{fig:consistency}
\end{figure}

To formalize this intuition we introduce the notion of \emph{type consistency}, defined in \autoref{fig:consistency}.
Two scalar types are consistent if they have a common inhabitant for some valid valuation of environment variables.
For function types, the relation is not symmetric:
a type $\funT{x}{T_x}{T}$ is consistent with a goal type if their return types are consistent for some value of $x$ of type $T_x$.

We add a premise $\env \vdash T\Consi T'$ to every rule in \autoref{fig:round-trip}
that already has the premise of the form $\env \vdash T\Subt T'$.
The additional premise has no effect on full applications, since for scalar types consistency is subsumed by subtyping.
The consistency check can, however, reject a \emph{partial} application $e$ allowed by subtyping,
due to goals generated by the rule \textsc{AppFO}, which have a vacuous argument type $\{B\mid\bot\}$.
It is easy to show that in the absence of consistency checks, 
any application of such $e$ would always be rejected by the subtyping check in \textsc{AppFO};
thus introducing consistency checks does not affect completeness of type checking.
With consistency checks in place, 
the term $\App{\T{append}}{xs}$ in the example above is rejected since the formula
$\App{\T{len}}{xs} = 6 \And \App{\T{len}}{r} \geq 0 \And \App{\T{len}}{\nu} = \App{\T{len}}{xs} + \App{\T{len}}{r} \And \App{\T{len}}{\nu} = 5$
is unsatisfiable.

\custompar{Liquid Abduction.}
Consider the \textsc{If} rule in \autoref{fig:round-trip}:
the type checker can analyze the two branches of the conditional independently of each other,
but can only proceed with either branch once the precise type of the guard has been inferred.
In the context of synthesis this amount to blindly enumerating type-correct boolean expressions as guards
and then checking if any of them enables synthesis of a correct branch.
The goal of this section is to improve the locality of the \textsc{If} rule
in order to avoid such blind enumeration.

The idea comes from \emph{condition abduction}~\cite{LeinoMi12,KneussKuKuSu13,AlurCR15}:
instead of starting with the guard, for which no information can be extracted from the goal type,
start by analyzing one of the branches and use logical abduction to infer the weakest assumption under which the branch fulfills the goal type.
If such a condition does not exist or is a contradiction, the branch candidate is deemed ill-typed;
otherwise the abduced condition can be used as a specification for the guard.

This strategy relies on the availability of a sufficiently fast mechanism to perform logical abduction, which is generally challenging.
In \lang, we treat unknown path conditions the same way as unknown refinements in polymorphic instantiations:
we restrict their valuations to liquid formulas over environment variables,
and use the greatest-fixpoint Horn solver (described in~\autoref{sec:theory:horn}) to discover the weakest such valuation.
We refer to the modified rule for conditionals as the \emph{liquid abduction} rule:
$$
\inference[\textsc{IfAb}]
{\env \entailsQ \psi & \mathsf{Sat}(\llbracket\Gamma\rrbracket_{\psi}\And\psi) \\ 
\env \entailsQ e \consume \{\T{Bool}\mid \nu = \psi\} \\
\env;\psi \entailsQ t_1 \consume T  &  \env;\neg\psi \entailsQ t_2 \consume T}
{\env \entailsQ\T{if}\ e\ \T{then}\ t_1\ \T{else}\ t_2 \consume T}
$$
This rule limits completeness of round-trip type checking by restricting valid guard types to the form above.
Most notably, it excludes guards that contain function composition,
and thus users have to provide wrapper components to encapsulate complex guard predicates;
in all our experiments, the set of required guard components was quite intuitive,
thus we conclude that the trade-off between expressiveness and efficiency offered by liquid abduction is reasonable in the context synthesis.

\custompar{Match Abduction.}
A similar technique can be used to propose pattern matching,
assuming the types of potential scrutinees are restricted to liquid types.
In this case, however, the liquid restriction imposes more substantial limitations on the structure of the program:
abduction only works if the scrutinee is a variable
and its datatype has at least one scalar constructor (such as \T{Nil} in \T{List}).
Thus, \lang employs a combined approach:
it first tries an abduction-based rule, but if that fails, the system reverts to the original \textsc{Match} rule of \autoref{fig:round-trip}.
As a result, type checking (and synthesis) enjoys the efficiency benefits of abduction without compromising completeness.

\subsection{The Local Liquid Type Checking Algorithm}\label{sec:theory:subtyping}

Starting from the round-trip typing rules presented above,
this section develops \emph{local liquid type checking}:
a deterministic algorithm that takes as input a \lang program $t$, an environment $\env$, a goal schema $S$, and a set of qualifiers $\quals$,
and either produces a derivation of $\env\entailsQ t\consume S$ or rejects the program.
The main challenge is to find suitable instantiations for polymorphic components, as required by the rule \textsc{Var$\forall$};
to this end, the algorithm replaces the type variables $\alpha_i$ in the component schema with fresh free type variables $\alpha'_i$%
\footnote{We prime the names of free type variables to differentiate them from the bound type variables of the top-level goal schema.},
extracts \emph{subtyping constraints} on $\alpha'_i$ from the subtyping premises of the derivation,
and then solves the subtyping constraints 
to either discover a valid type assignment mapping free type variables to liquid types, or conclude that such an assignment does not exist.

For the purpose of constraint solving, we extend the syntax of refinement terms with \emph{predicate unknowns} $P_i$.
Local liquid type checking maintains a set of subtyping constraints $\constraints=\{\env_i\vdash T_i\Subt T'_i\}$,
a set of Horn constraints $\horns=\{\psi_i\}$,
a \emph{type assignment} $\typeass=[\alpha'_i \mapsto T_i]$,
and a \emph{liquid assignment} $\liquidass=[P_i \to \{\psi\}_i]$.
We denote with $\valuation{\psi}{\liquidass}$ the formula $\psi$ with all predicate unknowns substituted with conjunctions of their valuations in $\liquidass$.
The type checking process alternates between the following two steps:
it either extends the type derivation by applying one of the rules of \autoref{fig:round-trip},
adding any of its subtyping premises to $\constraints$,
or it picks a constraint $c$ from $\constraints$ and solves it;
constraint solving is formalized in the procedure $\solve$ in \autoref{fig:solve}.

\begin{figure}
\small
\begin{align}
\nonumber\solve&(\env\vdash c) =\ \T{match}\ c\ \T{with}\\
\nonumber        &\mid  \{\alpha'\mid\psi\} \Subt T, \alpha'\in \dom(\typeass)  \longrightarrow\\
\label{fml:subst-l}      &\quad \constraints \gets \constraints \cup \{\env\vdash \refine(\typeass(\alpha'), \psi) \Subt T\}\\
\label{fml:subst-r}    &\mid  T \Subt \{\alpha'\mid\psi\}, \alpha'\in \dom(\typeass)  \longrightarrow\text{(symmetrical)}\\
\nonumber        &\mid  \{\alpha'\mid\psi_1\} \Subt \{\beta'\mid\psi_2\}  \longrightarrow\\ 
\label{fml:vars}      &\quad \constraints \gets \constraints \cup \{\alpha'\mid\psi_1\} \Subt \{\beta'\mid\psi_2\}\\
\nonumber        &\mid  \{\alpha'\mid\psi\} \Subt T, \alpha'\notin T  \longrightarrow \\
\nonumber        &\quad \typeass \gets \typeass[\alpha'\mapsto \fresh(T)];\\
\label{fml:fresh-l}    &\quad \constraints := \constraints\cup \{\env\vdash \{\alpha'\mid\psi\} \Subt T\} \\
\label{fml:fresh-r}    &\mid  T\Subt\{\alpha'\mid\psi\} \longrightarrow \text{(symmetrical)} \\
\nonumber        &\mid  (\funT{x}{T_x}{T_1}) \Subt (\funT{y}{T_y}{T_2}) \longrightarrow \\
\label{fml:fun} &\quad \constraints \gets \constraints \cup \{\env\vdash T_y \Subt T_x, \env;y: T_y\vdash [y/x]T_1\Subt T_2\} \\
\nonumber        &\mid  \{D\ T^i_1 \mid \psi_1\} \Subt \{D\ T^i_2 \mid \psi_2\} \longrightarrow \\
\label{fml:datatype} &\quad \constraints \gets \constraints \cup \{\env\vdash\{D \mid \psi_1\} \Subt \{D \mid \psi_2\}, \env\vdash T^i_1 \Subt T^i_2\} \\
\nonumber        &\mid  \{B \mid \psi_1\} \Subt \{B \mid \psi_2\}  \longrightarrow \\
\label{fml:base-invalid}  &\quad \horns \gets \horns\cup \{\llbracket\env\rrbracket_{\psi_1\Implies \psi_2}\And\psi_1\Implies \psi_2\};\\ 
\nonumber                 &\quad \liquidass \gets \horn(\liquidass,\horns) \\
            &\mid  \text{otherwise}  \longrightarrow \text{fail}
\end{align}
\vspace{-7mm}
\begin{align*}
&\refine(\{B \mid \psi\}, \psi') = \{B \mid \psi\And\psi'\}\\
&\fresh(T) 
      \begin{aligned}[t]
      &=\ \T{match}\ T\ \T{with}\\      
      &\mid \{\alpha'\mid\psi\}\longrightarrow \beta' \\ % \text{, where $\beta'$ fresh}\\
      &\mid\{D\ T^i \mid \psi\} \longrightarrow \{D\ \fresh(T^i) \mid P\},\quad\liquidass \gets \liquidass[P\mapsto \emptyset]\\ % \text{, where $P$ fresh}\\
      &\mid\{B \mid \psi\} \longrightarrow \{B \mid P\},\quad\liquidass \gets \liquidass[P\mapsto \emptyset]%\text{, where $P$ fresh}\\
      \end{aligned}\\
&\horn(\liquidass,\horns) = \T{if}\ \begin{aligned}[t]
                            &\forall h\in \horns. \mathsf{Valid}(\llbracket h\rrbracket_{\liquidass})\ \T{then}\ \liquidass\ \T{else}\\
                            &\T{let}\ h\gets \{\horns\mid \neg\mathsf{Valid}(\llbracket h\rrbracket_{\liquidass})\}\ \T{in}\\
                            &\T{let}\ \liquidass'\gets \strengthen(\liquidass, h)\ \T{in}\ \horn(\liquidass',\horns)
                            \end{aligned}
\end{align*}
\caption{Solving subtyping constraints.}\label{fig:solve}
\end{figure}

$\solve$ does one of the following, depending on the operands of a subtyping constraint:
it either substitutes a type variable for which an assignment already exists (\autoref{fml:subst-l}, \autoref{fml:subst-r}),
unifies a type variable with a type (\autoref{fml:fresh-l}, \autoref{fml:fresh-r}),
decomposes subtyping over compound types (\autoref{fml:fun}, \autoref{fml:datatype}),
or translates subtyping over scalar types into a Horn constraint 
and uses the procedure $\horn$, described in the next section, to find an $\liquidass$ that satisfies all Horn constraints (\autoref{fml:base-invalid}).
Local liquid type checking terminates when the entire type derivation has been built,
and all constraints in $\constraints$ are between free type variables (only \autoref{fml:vars} applies).

During unification of $\alpha'$ and $T$, procedure $\fresh$ inserts fresh predicate unknowns in place of all refinements in $T$;
note that due to the incremental nature of our algorithm,
$T$ might itself contain free type variables, which are simply replaced with fresh free type variables
to be unified later as more subtyping constraints arise.
This novel feature of local liquid type checking, which we call \emph{incremental unification}, 
is crucial for early error detection.
Existing refinement type checkers~\cite{Flanagan06,RondonKaJh08} cannot interleave shape and refinement discovery,
since they rely on the global Hindley-Milner inference algorithm to fully reconstruct the shapes of all types in the program
before discovering their refinements.

\begin{example*}
Starting from empty $\typeass$ and $\liquidass$,
$\solve(\vdash \alpha'\Subt \T{List}\ \beta'\mid\T{len}\ \nu > 0)$
instantiates $\alpha'$ by \autoref{fml:fresh-l} leading to $\typeass=[\alpha' \mapsto \{\T{List}\ \gamma'\mid P_0\}]$,
$\liquidass=[P_0\mapsto\emptyset]$ and recycles the subtyping constraint;
next by \autoref{fml:subst-l} and \autoref{fml:datatype}, the constraint is decomposed into
$\vdash \{\T{List}\ \mid P_0\}\Subt \{\T{List}\ \mid\T{len}\ \nu > 0\}$ and $\vdash\gamma' \Subt \beta'$.
The former produces a Horn constraint $P_0\Implies \T{len}\ \nu > 0$, which leads to strengthening $\liquidass[P_0]$,
while the latter is retained in $\constraints$.
If further type checking produces a subtyping constraint on $\beta'$, say $\T{Nat}\Subt\beta'$,
$\typeass$ will be extended with an assignment $[\beta'\to \{\T{Int}\mid P_1\}]$,
which in turn will lead to transforming the constraint on  $\gamma'$ into $\vdash\gamma' \Subt \{\T{Int}\mid P_1\}$
and instantiating $[\gamma'\to \{\T{Int}\mid P_2\}]$,
at which point all free type variables have been eliminated.
\end{example*}

\subsection{Solving Horn Clauses}\label{sec:theory:horn}

The set of Horn constraints $\horns$ produced by $\solve$ in \autoref{fig:solve}
consists of implications of the form $\psi\Implies\psi'$,
where each side is a conjunction of a known formula and zero or more predicate unknowns $P$.
The goal of the procedure $\horn$ is to find a liquid assignment to $P$ that validates all constraints in $\horns$
or determine that $\horns$ is unsatisfiable.
The space of possible valuations of each $P$ is $2^{\quals_P}$,
where $\quals_P$ is a set of atomic formulas
obtained by instantiating qualifiers $\quals$ in the environment where $P$ was created.

Local liquid type checking invokes $\horn$ after every new Horn constraint is issued,
and expects to detect an unsatisfiable set of constraints---and thus a type error---as early as possible.
Round-trip typing rules---%
in particular, \textsc{AppFO} and \textsc{If-Abd}---%
produce constraints in a specific order,
such that for each unknown $P$, 
implications where $P$ appears negatively (on the left) are issued before the ones where it appears positively (on the right).
To enable early error detection in this setting,
procedure $\horn$ looks for the \emph{weakest} valuation of each $P$ that validates all Horn constraints issued so far,
and deems $\horns$ unsatisfiable if for some $P$ such a valuation does not exist \emph{or is inconsistent}
(an inconsistent valuation can be safely discarded since it is guaranteed to violate some future constraint where $P$ appears positively).

As an optimization, $\horn$ always starts form the current assignment $\liquidass$ and possibly makes it stronger, 
since all weaker assignments are known to be too weak to satisfy the previously issued constraints
(for a fresh $P$, $\liquidass[P]$ is initialized with $\emptyset$). 
$\horn$ uses an iterative greatest-fixpoint computation, outlined in \autoref{fig:solve};
in every iteration, $\strengthen(\liquidass,\psi\Implies\psi')$ produces the weakest consistent assignment $\liquidass'$ 
strictly stronger than $\liquidass$,
such that $\valuation{\psi}{\liquidass'}\Implies \valuation{\psi'}{\liquidass}$ is valid
(or fails if this is not possible).
In general, $\liquidass'$ is not unique;
in this case our algorithm simply explores all alternatives independently,
which happens rarely enough in the context of refinement type checking and synthesis.

Implementing $\strengthen$ efficiently is challenging:
for every unknown $P$ in $\psi$,
the algorithm has to find the smallest subset of atomic predicates from 
$\quals_{P}\setminus\liquidass[P]$ that validates the implication.
Existing greatest-fixpoint Horn solvers~\cite{SrivastavaGu09} use breadth-first search,
which is exponential in the cumulative size of $\quals_P$
and does not scale sufficiently well to practical cases of condition abduction (see \autoref{sec:evaluation}).
Instead, we observe that this task is similar to the problem of finding minimal unsatisfiable subsets (MUSs) of a a set of formulas;
based on this observation, we build a practical algorithm for $\strengthen$ which we dub \textsc{MUSFix}.

The task of $\strengthen$ amounts to finding all MUSs of the set
$\bigcup_{\kappa\in\psi}(\quals_{\kappa}\setminus\liquidass[\kappa]) \cup \{\neg\valuation{\psi'}{\liquidass}\}$
under the assumption $\valuation{\psi}{\liquidass}$.
\textsc{MUSFix} borrows the main insight of the \textsc{Marco} algorithm~\cite{LiffitonPrMaMa15} for MUS enumeration,
which relies on the ability of the SMT solver to produce unsatisfiable cores from proofs.
We modify \textsc{Marco} to only produce MUSs 
that contain the negated right-hand side of the Horn constraint, $\neg\valuation{\psi'}{\liquidass}$,
since $\horn$ should only produce consistent solutions.
For each resulting MUS (stripped of $\neg\valuation{\psi'}{\liquidass}$),
\textsc{MUSFix} finds all possible partitions into valuations of individual predicate unknowns.
Since MUSes are normally much smaller than the original set of formulas, 
a straightforward partitioning algorithm works well and rarely yield more than one valid partition.
As an important optimization,
when MUS enumeration returns multiple syntactically minimal subsets,
\textsc{MUSFix} prunes out \emph{semantically} redundant subsets,
\ie it removes a subset $m_i$ if $\bigwedge m_i \Implies \bigwedge m_j$ for some $j \neq i$.

\subsection{Synthesis from Refinement Types}\label{sec:theory:synthesis}

From the rules of round-trip type checking we can obtain synthesis rules, following the approach of~\cite{OseraZd15}
and reinterpreting the checking and strengthening judgments
in such a way that the term $t$ is considered unknown.
This interpretation yields a synthesis procedure,
which, given a goal schema $S$, picks a rule where the goal schema in the conclusion matches $S$,
and constructs the term $t$ from subterms obtained from the rule's premises.
More concretely, starting from the top-level goal schema $S$,
the algorithm always starts by applying rule \textsc{Fix} (if $S^\before$ is defined)
followed by \textsc{TAbs} (if the schema is polymorphic),
and finally \textsc{Abs} (if the goal type is a function type).
Given a scalar goal,
the procedure performs exhaustive enumeration of well-typed E-terms up to a given bound on their depth,
solving subtyping constraints at every node and
simultaneously abducing a path condition as per the \textsc{If-Abd} rule.
If the resulting condition $\psi$ is trivially true, the algorithm has found a solution;
if $\psi$ is inconsistent, the E-term is discarded;
otherwise, the algorithm generates a conditional 
and proceeds to synthesize its remaining branch under the fixed assumption $\neg\psi$,
as well a term of type $\{\T{Bool}\mid \nu=\psi\}$ to be used as the branch guard.
Once all possible E-terms are exhausted,
the algorithm attempts to synthesize a pattern match using an arbitrary E-term as a scrutinee,
unless the maximal nesting depth of matches has been reached.

\custompar{Soundness and Completeness.}
Soundness of synthesis follows straightforwardly from soundness of round-trip type checking, 
since each program candidate is constructed together with a typing derivation in the round-trip system.
Completeness is less obvious:
due to condition abduction, the synthesis procedure only explores programs
where the left branch of each conditional is an E-term.
We can show that every \lang program can be rewritten to have this form
(by flattening nested conditionals and pushing conditionals inside matches).
Thus the synthesis procedure is complete in the following sense: 
for each schema $S$, if there exists a term $t$,
such that the depth of applications and pattern matches in $t$ are within the given bounds,
the procedure is guarantees to find some term $t'$ that also type-checks against $S$;
if such a term $t$ does not exist, the procedure will terminate with a failure.
Note that the algorithm imposes no a-priori bound on the nesting depth of conditionals
(which is crucial for completeness as stated above);
this does not preclude termination,
since in any given environment, liquid formulas partition the input space into finitely many parts,
and every condition abduction is guaranteed to cover a nonempty subset of these parts.

\section{Evaluation}\label{sec:evaluation}

\begin{table*}[!htbp]
\begin{center}
\scriptsize
\resizebox{\textwidth}{!}{
\begin{tabular}{@{} r|c| cccc | cccccc @{}}
\head{Group} & \head{Description}  & \head{\#G} & \head{Components} & \head{\#M} & \head{Spec} & \head{Code} & \head{T-all} & \head{T-def} & \head{T-nrt} & \head{T-ncc} & \head{T-nmus} \\	

\hhline{============}

\multirow{20}{*}{\parbox{1cm}{\vspace{-0.85\baselineskip}\center{List}}} & is empty & 1 & true, false & 1 & 6 & 6 & 0.02 & 0.02 & 0.02 & 0.02 & 0.01 \\
 & is member & 1 & true, false, $=$, $\neq$ & 2 & 6 & 18 & 0.11 & 0.11 & 0.13 & 0.10 & - \\
 & duplicate each element & 1 &  & 1 & 7 & 16 & 0.05 & 0.05 & - & 0.08 & 0.04 \\
 & replicate & 1 & 0, inc, dec, $\leq$, $\neq$ & 1 & 4 & 21 & 0.05 & 0.05 & 9.63 & 0.05 & - \\
 & append two lists & 1 &  & 1 & 8 & 15 & 0.15 & 0.09 & - & 0.13 & 0.10 \\
 & concatenate list of lists & 1 & append & 3 & 5 & 12 & 0.05 & 0.05 & 0.22 & 0.04 & 0.04 \\
 & take first $n$ elements & 1 & 0, inc, dec, $\leq$, $\neq$ & 1 & 8 & 27 & 0.12 & 0.12 & 55.82 & 0.12 & - \\
 & drop first $n$ elements & 1 & 0, inc, dec, $\leq$, $\neq$ & 1 & 11 & 20 & 0.10 & 0.10 & 7.87 & 0.09 & - \\
 & delete value & 1 & $=$, $\neq$ & 2 & 8 & 26 & 0.10 & 0.10 & 0.17 & 0.12 & - \\
 & map & 1 &  & 1 & 5 & 22 & 0.03 & 0.03 & 0.06 & 0.03 & 0.02 \\
 & zip & 1 &  & 1 & 10 & 22 & 0.08 & 0.08 & - & 0.10 & 0.07 \\
 & zip with function & 1 &  & 1 & 10 & 33 & 0.07 & 0.07 & - & 0.17 & 0.06 \\
 & cartesian product & 1 & append, map & 3 & 8 & 26 & 0.30 & 0.29 & 5.83 & 0.25 & 0.23 \\
 & $i$-th element & 1 & 0, inc, dec, $\leq$, $\neq$ & 1 & 12 & 20 & 0.05 & 0.05 & 0.38 & 0.05 & - \\
 & index of element & 1 & 0, inc, dec, $=$, $\neq$ & 2 & 8 & 20 & 0.08 & 0.08 & 0.14 & 0.07 & - \\
 & insert at end & 1 &  & 2 & 21 & 19 & 0.10 & 0.10 & 0.24 & 0.11 & 0.12 \\
 & reverse & 1 & insert at end & 2 & 15 & 12 & 0.09 & 0.10 & 0.29 & 0.12 & 0.09 \\
 & foldr & 1 &  & 2 & 14 & 32 & 0.10 & 0.10 & - & 0.10 & 0.44 \\
 & length using fold & 1 & 0, inc, dec & 2 & 4 & 17 & 0.03 & 0.07 & 0.03 & 0.03 & 0.02 \\
 & append using fold & 1 &  & 2 & 8 & 20 & 0.04 & 2.19 & 0.05 & 0.04 & 0.03 \\
\hline\multirow{5}{*}{\parbox{1cm}{\vspace{-0.85\baselineskip}\center{Unique list}}} & insert & 1 & $=$, $\neq$ & 2 & 8 & 26 & 0.27 & 0.22 & 0.85 & 0.20 & - \\
 & delete & 1 & $=$, $\neq$ & 2 & 8 & 22 & 0.18 & 0.19 & 1.07 & 0.26 & - \\
 & remove duplicates & 2 & is member & 2 & 13 & 47 & 0.36 & 0.87 & 0.72 & 0.33 & - \\
 & remove adjacent dupl. & 1 & $=$, $\neq$ & 3 & 5 & 32 & 1.33 & 1.32 & - & 1.31 & - \\
 & integer range & 1 & 0, inc, dec, $\leq$, $\neq$ & 2 & 13 & 23 & 2.36 & 2.33 & 22.27 & 2.33 & - \\
\hline\multirow{3}{*}{\parbox{1cm}{\vspace{-0.85\baselineskip}\center{Strictly sorted list}}} & insert & 1 & $<$ & 2 & 8 & 41 & 0.18 & 0.17 & 0.43 & 0.16 & - \\
 & delete & 1 & $<$ & 2 & 8 & 29 & 0.10 & 0.09 & 0.21 & 0.10 & - \\
 & intersect & 1 & $<$ & 2 & 8 & 40 & 0.33 & 0.32 & 0.68 & 0.34 & - \\
\hline\multirow{11}{*}{\parbox{1cm}{\vspace{-0.85\baselineskip}\center{Sorting}}} & insert (sorted) & 1 & $\leq$, $\neq$ & 2 & 8 & 34 & 0.25 & 0.24 & 0.68 & 0.23 & - \\
 & insertion sort & 1 & insert (sorted) & 4 & 5 & 12 & 0.06 & 0.06 & 0.20 & 0.06 & 0.05 \\
 & sort by folding & 1 & foldr, $\leq$, $\neq$ & 2 & 11 & 47 & 2.14 & - & - & 2.21 & - \\
 & extract minimum & 1 & $\leq$, $\neq$ & 4 & 23 & 40 & 4.28 & 4.35 & - & 7.58 & - \\
 & selection sort & 1 & extract minimum & 6 & 5 & 16 & 0.49 & 0.44 & - & 0.42 & 0.38 \\
 & balanced split & 1 &  & 4 & 31 & 33 & 0.96 & 0.51 & - & 1.40 & 0.80 \\
 & merge & 1 & $\leq$, $\neq$ & 2 & 17 & 41 & 2.19 & 14.61 & - & 6.85 & - \\
 & merge sort & 1 & split, merge & 6 & 11 & 25 & 2.10 & 2.10 & - & 2.52 & 1.69 \\
 & partition & 1 & $\leq$ & 4 & 27 & 40 & 2.84 & 7.89 & - & 3.42 & - \\
 & append with pivot & 1 &  & 2 & 28 & 22 & 0.22 & 0.15 & 0.58 & 0.22 & 0.19 \\
 & quick sort & 1 & partition, append w/pivot & 6 & 11 & 22 & 2.71 & 18.45 & - & 2.49 & 4.94 \\
\hline\multirow{4}{*}{\parbox{1cm}{\vspace{-0.85\baselineskip}\center{Tree}}} & is member & 1 & false, not, or, $=$ & 2 & 6 & 28 & 0.29 & 0.29 & 7.90 & 0.28 & - \\
 & node count & 1 & 0, 1, + & 1 & 4 & 18 & 0.20 & 0.20 & - & 0.91 & 0.14 \\
 & preorder & 1 & append & 2 & 5 & 18 & 0.21 & 0.20 & - & 0.91 & 0.15 \\
 & create balanced & 1 & 0, inc, dec, $\leq$, $\neq$ & 2 & 7 & 29 & 0.14 & 0.15 & - & 0.21 & - \\
\hline\multirow{4}{*}{\parbox{1cm}{\vspace{-0.85\baselineskip}\center{BST}}} & is member & 1 & true, false, $\leq$, $\neq$ & 2 & 6 & 37 & 0.09 & 0.08 & 0.10 & 0.08 & - \\
 & insert & 1 & $\leq$, $\neq$ & 2 & 8 & 55 & 0.91 & 0.88 & - & 0.82 & - \\
 & delete & 1 & $\leq$, $\neq$ & 2 & 8 & 68 & 5.68 & 5.62 & - & 10.74 & - \\
 & BST sort & 5 & $\leq$, $\neq$ & 6 & 51 & 115 & 1.38 & 1.35 & - & 1.25 & - \\
\hline\multirow{5}{*}{\parbox{1cm}{\vspace{-0.85\baselineskip}\center{Binary Heap}}} & is member & 1 & false, not, or, $\leq$, $\neq$ & 2 & 6 & 43 & 0.38 & 0.38 & 9.63 & 0.35 & - \\
 & insert & 1 & $\leq$, $\neq$ & 2 & 8 & 55 & 0.51 & 0.50 & 8.83 & 0.48 & - \\
 & 1-element constructor & 1 & $\leq$, $\neq$ & 2 & 5 & 8 & 0.02 & 0.02 & 0.02 & 0.02 & 0.02 \\
 & 2-element constructor & 1 & $\leq$, $\neq$ & 2 & 6 & 55 & 0.08 & 0.08 & 0.25 & 0.07 & - \\
 & 3-element constructor & 1 & $\leq$, $\neq$ & 2 & 7 & 246 & 2.10 & 2.12 & - & 1.98 & - \\
\hline\multirow{6}{*}{\parbox{1cm}{\vspace{-0.85\baselineskip}\center{AVL}}} & rotate left & 3 & inc & 3 & 104 & 91 & 11.08 & 12.43 & - & 17.06 & 10.08 \\
 & rotate right & 3 & inc & 3 & 107 & 91 & 19.23 & 18.34 & - & 36.35 & 17.87 \\
 & balance & 1 & rotate, nodeHeight, isSkewed, isLHeavy, isRHeavy & 4 & 31 & 119 & 1.56 & - & - & 1.76 & - \\
 & insert & 1 & balance, $<$ & 3 & 22 & 47 & 1.84 & 1.81 & - & 1.64 & - \\
 & extract minimum & 1 & $<$ & 5 & 11 & 25 & 1.92 & 1.87 & - & 1.72 & - \\
 & delete & 2 & extract minimum, balance, $<$ & 5 & 37 & 63 & 15.67 & - & - & 13.79 & - \\
\hline\multirow{3}{*}{\parbox{1cm}{\vspace{-0.85\baselineskip}\center{RBT}}} & balance left & 2 &  & 9 & 143 & 137 & 5.62 & 5.53 & - & 48.47 & - \\
 & balance right & 2 &  & 9 & 144 & 137 & 7.63 & 7.72 & - & 45.32 & - \\
 & insert & 3 & balance left, right, $\leq$, $\neq$ & 9 & 49 & 112 & 8.95 & 8.53 & - & 7.93 & - \\
\hline\multirow{3}{*}{\parbox{1cm}{\vspace{-0.85\baselineskip}\center{User}}} & desugar AST & 1 & 0, 1, 2 & 4 & 5 & 46 & 1.17 & 1.10 & - & 1.23 & 0.78 \\
 & make address book & 1 & is private & 3 & 5 & 35 & 0.62 & 3.67 & - & 0.94 & 0.55 \\
 & merge address books & 1 & append & 3 & 8 & 19 & 0.35 & 5.85 & - & 0.31 & 0.24 \\
\hline
\end{tabular}
}
\end{center}
\caption{
Benchmarks and \tool results.
For each benchmark, we report the number of synthesis goals \head{\#G};
the set of provided \head{Components};
the number of defined measures \head{\#M};
cumulative size of \head{Spec}ification and synthesized \head{Code} (in AST nodes) for all goals;
as well as \tool running times (in seconds) with minimal bounds (\head{T-all}),
with default bounds (\head{T-def}),
without round-trip checking (\head{T-nrt}),
without type consistency checking (\head{T-ncc}),
and without \textsc{MUSFix} (\head{T-nmus}).
``-'' denotes timeout of 2 minutes or out of memory.
}
\label{fig:evaluation}
\end{table*}

We performed an extensive experimental evaluation of \tool
with the goal of assessing usability and scalability of the proposed synthesis technique
compared to existing alternatives.
This goal materializes into the following research questions:
\begin{enumerate}[(1)]
\item Are refinement types supported by \tool \emph{expressive} enough to specify interesting programs, including benchmarks proposed in prior work?
\item How \emph{concise} are \tool's input specifications compared both to the synthesized solutions and to inputs required by existing techniques?
\item Are \tool's inputs \emph{intuitive}: in particular, is the algorithm applicable to specifications not tailored for synthesis?
\item How \emph{scalable} is \tool: can it handle benchmarks tackled by existing synthesizers?
Can it scale to more complex programs than those previously reported in the literature?
\item How is synthesis performance impacted by various features of \tool and its type system?
\end{enumerate}

\subsection{Benchmarks}\label{sec:eval:benchmarks}

In order to answer the research questions stated above, we arranged a benchmark suite that consists of \exCount synthesis challenges
from various sources, representing a range of problem domains.
In the interest of direct comparison with existing synthesis tools, 
our suite includes benchmarks that had been used in the evaluation of those tools%
~\cite{KneussKuKuSu13,LeinoMi12,OseraZd15,FrankleOWZ16,FeserChDi15,AlbarghouthiGuKi13,MilicevicNKJ15,AlurCR15}. % ,InalaQLS15};
From each of these papers, we picked top three most complex challenges (judging by the reported synthesis times) that were expressible in \tool's refinement logic,
plus several easier problems that were common or particularly interesting. % Make this more precise

Our second source of benchmarks are verification case studies from the LiquidHaskell tutorial~\cite{LiquidHaskellTutorial}.
The purpose of this second category is two-fold:
first, these problems are larger and more complex than existing synthesis benchmarks,
and thus can show whether \tool goes beyond the state of the art in synthesis;
second, the specifications for these problems have been written by independent researchers and for a different purpose,
and thus can serve as evidence that input accepted by \tool is sufficiently general and intuitive.
Out of the total of 14 case studies, we picked 5 that came with sufficiently strong functional specifications
(list sorting, binary-search trees, content-aware lists, unique lists, and AVL trees),
erased all implementations,
and made relatively straightforward syntactic changes in order to obtain valid \tool input.

\autoref{fig:evaluation} lists the \exCount benchmarks together with  % their sources and
some metrics of our type-based specifications: 
the number of synthesis goals including auxiliary functions,
the set of components provided,
the number of measures used,
and the cumulative size of refinements.
Note that the reported specification size only includes refinements in the signatures of the synthesis goals;
refinements in component functions are excluded since every such function (except trivial arithmetic operations and helper functions) serves as a synthesis goal in another benchmark;
refinements in datatype definitions are also excluded, since those definitions are reusable between all benchmarks in the same problem domain.
Full specifications are available from the \tool repository~\cite{SynquidRepo}.

The benchmarks are drawn from a variety of problem domains with the goal of exercising different features in \tool.
List and tree benchmarks demonstrate pattern matching, structural recursion, 
the ability to generate and use polymorphic and higher-order functions (such as \T{map} and \T{fold}), 
as well as reasoning about nontrivial properties of data structures,
both universal (e.g.\ all elements are non-negative) and recursive (e.g.\ size and set of elements).
Our most advanced benchmarks include sorting and operations over data structures with complex representation invariants,
such as binary search trees, heaps, and balanced trees.
These benchmarks showcase expressiveness of refinement types,
exercise \tool's ability to perform nontrivial reasoning through refinement discovery,
and represent a scalability challenge beyond the current state of the art in synthesis.
Finally, we included several benchmarks operating on ``custom'' datatypes
(including the ``address book'' case study from~\cite{KneussKuKuSu13})
in order to demonstrate that \tool's applicability is not limited to standard textbook examples.

\subsection{Results}\label{sec:eval:individual}

Evaluation results are summarized in \autoref{fig:evaluation}.
\tool was able to synthesize (and fully verify) solutions for all \exCount benchmarks;
the table lists sizes of these solutions in AST nodes (\emph{Code}) 
as well as synthesis times in seconds (\emph{T-all}).

The results demonstrate that \tool is efficient in synthesizing a variety of programs:
all but 7 benchmarks are synthesized within 5 seconds; 
it also scales to programs of nontrivial size, 
including complex recursive (red-black tree insertion of size 69) 
and non-recursive functions (3-value binary heap constructor of size 246).
Even though specification sizes for some benchmarks are comparable with the size of the synthesized code, 
for many complex problems the benefits of describing computations as refinement types are significant:
for example, the type-based specifications of the three main operations on binary-search trees are over six times more concise than their implementations.

The synthesis times discussed above were obtained for optimal exploration bounds, which could differ across benchmarks.
\autoref{fig:evaluation} also reports synthesis times \emph{T-def} for a setting
where all benchmarks in the same category share the same exploration bounds. 
Although this inevitably slows down synthesis, on most of the benchmarks the performance penalties were not drastic: 
only three benchmarks failed to terminate within the two-minute timeout.

In order to assess the impact on performance of various aspects of our algorithm and implementation,
\autoref{fig:evaluation} reports synthesis times using three variants of \tool, where certain features were disabled:
the column \emph{T-nrt} corresponds to replacing round-trip type checking with bidirectional type checking
(that is, disabling subtyping checks for partial applications);
\emph{T-ncc} corresponds to disabling type consistency checks;
\emph{T-nmus} corresponds to replacing \textsc{MUSFix} with naive breadth-first search.
The results demonstrate that the most significant contribution comes from using \textsc{MUSFix}:
without this feature 37 out of \exCount benchmarks time out,
since breadth-first search cannot handle condition abduction even with a moderate number of logical qualifiers.
The second most significant feature is round-trip type checking, with 33 benchmarks timing out when disabled,
while consistency checks only bring significant speedups for the most complex examples.

\subsection{Comparative Evaluation}\label{sec:eval:comparative}

\begin{table}
\begin{center}

\resizebox{\columnwidth}{!}{
\begin{tabular}{@{} r | c | cccc @{}}
 & \head{Benchmark} & \head{Spec} & \head{SpecS} & \head{Time} & \head{TimeS} \\

\hhline{======}

\multirow{3}{*}{\rotatebox{90}{\textsc{Leon}}}&strict sorted list delete&14&8&15.1&0.10 \\
&strict sorted list insert&14&8&14.1&0.18 \\
&merge sort&9&11&14.3&2.1 \\
\hline \multirow{3}{*}{\rotatebox{90}{\textsc{Jen}}}&BST find&51&6&64.8&0.09 \\
&bin. heap 1-element&80&5&61.6&0.02 \\
&bin. heap find&76&6&51.9&0.38 \\
\hhline{======} \multirow{3}{*}{\rotatebox{90}{\textsc{Myth}}}&sorted list insert&12&8&0.12&0.25 \\
&list rm adjacent dupl.&13&5&0.07&1.33 \\
&BST insert&20&8&0.37&0.91 \\
\hline \multirow{3}{*}{\rotatebox{90}{$\lambda^2$}}&list remove duplicates&7&13&231&0.36 \\
&list drop&6&11&316.4&0.1 \\
&tree find&12&6&4.7&0.29 \\
\hline \multirow{3}{*}{\rotatebox{90}{\textsc{Esc}}}&list rm adjacent dupl.&n/a&5&1&1.33 \\
&tree create balanced&n/a&7&0.24&0.14 \\
&list duplicate each&n/a&7&0.16&0.05 \\
\hline \multirow{3}{*}{\rotatebox{90}{\textsc{Myth2}}}&BST insert&n/a&8&1.81&0.91 \\
&sorted list insert&n/a&8&1.02&0.25 \\
&tree count nodes&n/a&4&0.45&0.20 \\

\end{tabular}
}

\end{center}
\caption{
\small{
Comparison to other synthesizers.
For each benchmark we report:
\head{Spec}, specification size (or the number of input-output examples) for respective tool;
\head{SpecS}, specification size for \tool (from \autoref{fig:evaluation});
\head{Time}, reported running time for respective tool;
\head{TimeS}, running time for \tool (from \autoref{fig:evaluation}).
}
}
\label{fig:comparative}
\end{table}

We compared \tool with state-of-the-art synthesis tools that target recursive functional programs
and offer a comparable level of automation.
The results are summarized in \autoref{fig:comparative}.
For each tool, we list the three most complex benchmarks reported in the respective paper that were expressible in \tool's refinement logic;
for each of the three benchmarks we report the specification size (if available) and the synthesis time;
for ease of comparison, we repeat the same two metrics for \tool (copied over from \autoref{fig:evaluation}).
Note that the synthesis times are not directly comparable, 
since the results for other tools are taken from respective papers
and were obtained on different hardware;
however, the differences of an order of magnitude or more are still significant, 
since they can hardly be explained by improvements in single-core hardware performance.

We split the tools into two categories according to the specification and verification mechanism they rely on.

\begin{figure*}[!ht]
% \large
\centering
\resizebox{!}{6cm}{
\includegraphics{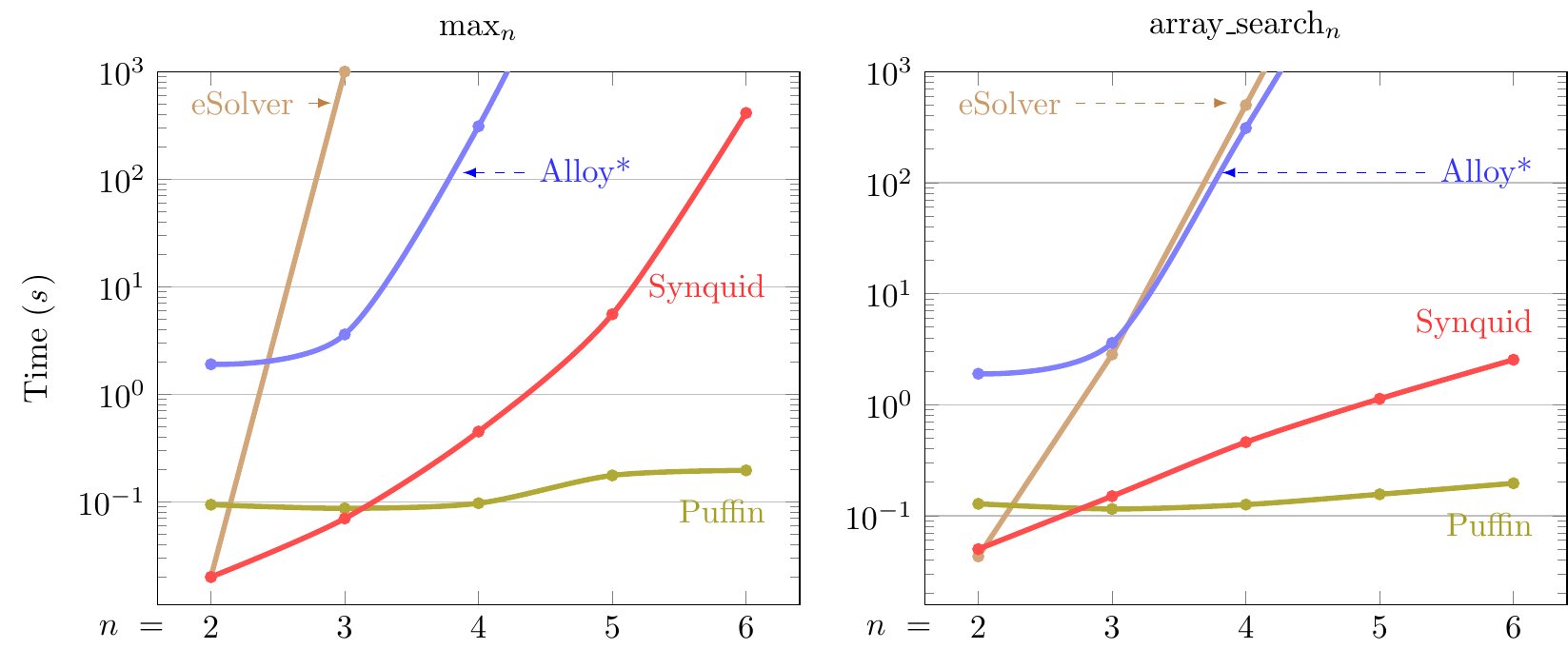}
}
% \hspace{5mm}
\caption{Evaluation on non-recursive benchmarks.}
\label{fig:arithmetic_graphs}
\end{figure*}

\custompar{Formal Specifications with Deductive Verification.}
The first category includes \textsc{Leon}~\cite{KneussKuKuSu13} and \textsc{Jennisys}~\cite{LeinoMi12};
both tools use pre- and post-conditions (and data structure invariants) to describe computations,
and rely on unbounded, SMT-based verification to validate candidate programs
(and thus provide the same correctness guarantees as \tool).
Unlike \textsc{Leon} and \tool, \textsc{Jennisys} targets imperative, heap-based programs;
the evaluation in~\cite{LeinoMi12}, however, focuses on side-effect free benchmarks.
Both tools use variants of condition abduction, which makes their exploration strategies similar to \tool's.

For both tools, translating their three most complex benchmarks into \tool proved to be straightforward.
This suggests that our decidable refinement logic is not too limiting in practice,
compared to other formal specification languages used for synthesis.
Our specifications are on average slightly more concise than \textsc{Leon}'s and significantly more concise than those in \textsc{Jennisys};
the latter is largely due to the heap-based language,
but the results still indicate that embedding predicates into types can help curb the verboseness of traditional Hoare-style specifications. 

\tool was able to synthesize solutions to all problems tackled by the other two tools in this category.
The converse is not true: automatic verification of some of \tool's benchmarks (such as the binary-search tree example in \autoref{sec:overview}) 
requires invariant discovery, which is not supported by the other two tools. 
This suggests that \tool \emph{qualitatively} differs from other state-of-the-art synthesizers in terms of the class of programs for which a verified solution can be synthesized.
On the benchmarks where the other tools are applicable, \tool demonstrates considerably smaller running times,
which suggests that fast verification and early pruning enabled by type-based specifications indeed improve the scalability of synthesis. 

\custompar{Input/Output Examples.}
Our second category of tools includes \textsc{Myth}~\cite{OseraZd15}, \textsc{$\lambda^2$}~\cite{FeserChDi15} and \textsc{Escher}~\cite{AlbarghouthiGuKi13},
which synthesize programs from concrete input-output examples,
as well as \textsc{Myth2}~\cite{FrankleOWZ16}, which uses generalized input-output examples.
Using refinement types, we were able to express 3 out of 3, 10, 5, and 7 of their most complex benchmarks, receptively.
The functions we failed to specify either manipulate nested structures in a representation-specific way (such as ``insert a tree under each leaf of another tree''),
or perform filtering (``list of nodes in a tree that match a predicate'').

At the same time, we found cases where refinement types are concise and intuitive, 
while providing input-output examples is extremely tedious.
One of those cases is insertion into a binary search tree:
\textsc{Myth} requires 20 examples,
each of which contains two bulky tree instances and has to define the precise position where the new element is to be inserted;
the type-based specification for this problem, given in \autoref{sec:overview}, 
is straightforward and only defines the abstract effect of the operation relevant to the user. 
This suggests that in general, logic-based specification techniques, including refinement types in \tool, 
are a better fit for describing operations that maintain a complex representation invariant but have a simple abstract effect,
while example-based approaches fare better when describing operations that inherently expose the complex representation of a data structure.

Experiments with example-based tools only report the number of examples required for synthesis and not their sizes;
however, we can safely assume that each example contains multiple AST nodes,
and thus conclude that type-based specifications for the benchmarks in \autoref{fig:comparative} are more concise.
By imposing more constraints on the set of examples (such as \emph{trace completeness}~\cite{OseraZd15}) and increasing its size,
example-based synthesizers can trade off user effort for synthesis time.
On the benchmarks under comparison, \textsc{Myth} appears to favor performance, while $\lambda^2$ prefers smaller example sets.
\tool tries to offer the best of both world and achieves good performance with concise specifications.

\subsection{Evaluation on Non-recursive Benchmarks}

In order to asses the scalability of \textsc{MUSFix} on larger search spaces, 
we evaluated \tool on two parametrized benchmarks from the linear integer arithmetic track of the SyGuS'14 competition~\cite{AlurBJMRSSSTU13}: 
$\T{max}_n$ (find maximum of $n$ integer arguments)
and $\T{array_search}_n$ (find the position of a value in a sorted array with $n$ elements).
Both benchmarks target non-recursive programs that consist of a series of nested conditionals;
moreover, the search space for the branch guards grows exponentially with $n$.
This makes the two problems ideal benchmarks for condition abduction techniques.

\autoref{fig:arithmetic_graphs} shows \tool synthesis times on the two benchmarks for $n = 2,3,\ldots,6$.
For reference, we also plot the results for the enumerative solver (the fastest of the SyGuS baseline solvers),
as well as the higher-order solver \textsc{Alloy*}~\cite{MilicevicNKJ15},
and \textsc{Puffin}~\cite{AlurCR15}, a specialized synthesizer for conditional integer-arithmetic expressions%
\footnote{The results for these tools are taken from their respective papers;
only differences in the order of magnitude are significant.}.
The results show that \tool's condition abduction scales relatively well compared to general synthesizers,
but loses to \textsc{Puffin}'s theory-specific abduction engine.

\section{Related Work}\label{sec:related}

Our work is the first to leverage general decidable refinement types for synthesis, but it builds on a number of ideas
from prior work as has been highlighted already throughout the paper. Specifically, our work combines ideas from 
two areas: synthesis of recursive functional programs and refinement type inference.

\custompar{Synthesis of Recursive Functional Programs.}
A number of recent systems target recursive functional programs and use type information in some form to restrict the search space. 
The most closely related to our work are \textsc{Myth}~\cite{OseraZd15}, \textsc{Myth2}~\cite{FrankleOWZ16}, and \textsc{Leon}~\cite{KneussKuKuSu13}.

\textsc{Myth} pioneered the idea of leveraging bidirectional type checking for synthesis. 
However, \textsc{Myth} does not support polymorphism or refinement types. 
Instead, the system relies on examples in order to specify the desired functionality. 
For certain functions, providing examples is easy
whereas writing a refinement type is cumbersome or, due to the limitations of decidable refinement logic, even impossible.
That said, examples in general do not fully specify a program;
thus programming by example always involves a manual verification step. 
Moreover, for some less intuitive problems, such as insertion into a balanced tree or AST transformations,
providing input-output examples requires familiarity with all details and corner cases of the algorithm,
whereas refinement types enable a more abstract specification.
Additionally, \textsc{Myth} expects the set of examples to be \emph{trace complete}, 
which means that for any example the user provides, there should also be examples corresponding to any recursive calls made on that input. 
Other systems that use a combination of types and input-output examples,
and thus have similar advantages and disadvantages relative to our system,
include $\lambda^2$~\cite{FeserChDi15} and \textsc{Escher}~\cite{AlbarghouthiGuKi13}.

\textsc{Myth2} generalizes example-based synthesis: 
it treats examples as singleton types, and extends the input language with intersection and union types, as well as parametric polymorphism.
This addresses some of the shortcomings of concrete input-output examples (in particular, their verboseness),
however, in the absence of dependent function types most interesting programs still cannot be fully speficied.
Combining \lang's dependent types with singletons, intersection, and unions found in \textsc{Myth2}
is an interesting direction for future work.

In \textsc{Leon}, synthesis problems are defined by first-order specifications with recursive predicates,
and verification is based on a semi-decision procedure~\cite{SuterKoKu11}, implemented on top of an SMT solver.
\textsc{Leon}'s verification engine does not support invariant inference,
which prevents it from generating provably correct implementations for problems such as insertion into a sorted list or a binary search tree.
The general synthesis strategy is similar to ours:
first decompose the specification 
and then switch to generate-and-check mode, enhanced with condition abduction. 
Unlike our system, \textsc{Leon} does not perform systematic specification decomposition in the generate-and-check mode,
and lacks support for polymorphism and high-order functions.

The use of type information has also proved extremely useful for code completion~\cite{MandelinXuBoKi05,PerelmanGuBaGr12,GveroKuKuPi13}, 
although none of these systems rely on a type system as expressive as ours, 
and they are designed for a very different set of tradeoffs compared to \tool. 
For example, because the problem is highly under-constrained, these systems place significant emphasis on the ranking of solutions.

Another important body of related work related is \emph{hole driven development}, 
as embodied in systems like Agda~\cite{Norell09} and Idris~\cite{Brady13}, 
which leverage a rich type system to aid development, 
but are meant to be used interactively rather than to perform complete synthesis. 
Djinn~\cite{Augustsson14} serves a similar purpose but uses the less expressive Haskell type system.

\custompar{Refinement Type Checking.}
Our type checking algorithm is based on liquid type inference~\cite{RondonKaJh08,KawaguchiRoJh09,VazouRoJh13,VazouSeJh14,VazouSeJhVyJo14},
which pioneered combining Hindley-Miler unification with predicate abstraction.
We integrate their ideas with bidirectional type checking~\cite{PierceTu00}, which has been used before both 
for other flavors of refinement types~\cite{XiPf99,DaviesPf00,DunfieldPf04} and for unrestricted dependent types~\cite{Coquand96},
but not for general decidable refinement types.
Another difference with liquid types is that we use greatest-fixpoint predicate abstraction procedure inspired by~\cite{SrivastavaGu09},
and improved using an algorithm for efficient MUS enumeration~\cite{LiffitonPrMaMa15}.

\custompar{Logical Abduction.}
The concept of abduction in logical reasoning has found numerous applications in programming languages,
including specification inference~\cite{DilligDLM13,AlbarghouthiDG16} and program synthesis~\cite{KneussKuKuSu13,LeinoMi12,AlurCR15}.
Abduction techniques based on quantifier elimination~\cite{DilligDLM13,AlbarghouthiDG16}
and theory-specific unification operators~\cite{AlurCR15},
are precise and efficient, but only applicable to restricted domains.
\tool performs abduction using predicate abstraction and MUS enumeration,
which can be applied to a wider range of specification logics,
but its precision is limited to the given set of logical qualifiers.

\section*{Acknowledgements}
We thank Aleksandar Milicevic, Shachar Itzhaky, Ranjit Jhala, and the anonymous reviewers for their valuable input. 
This work was funded by DARPA Grant \textsl{FA8750-14-2-0242} (MUSE) and NSF Grants \textsl{CCF-1139056} and \textsl{CCF-1438969}.

\bibliographystyle{abbrvnat}
\bibliography{synquid}

\begin{thebibliography}{40}
\providecommand{\natexlab}[1]{#1}
\providecommand{\url}[1]{\texttt{#1}}
\expandafter\ifx\csname urlstyle\endcsname\relax
  \providecommand{\doi}[1]{doi: #1}\else
  \providecommand{\doi}{doi: \begingroup \urlstyle{rm}\Url}\fi

\bibitem[Albarghouthi et~al.(2013)Albarghouthi, Gulwani, and
  Kincaid]{AlbarghouthiGuKi13}
A.~Albarghouthi, S.~Gulwani, and Z.~Kincaid.
\newblock Recursive program synthesis.
\newblock In \emph{CAV}, 2013.

\bibitem[Albarghouthi et~al.(2016)Albarghouthi, Dillig, and
  Gurfinkel]{AlbarghouthiDG16}
A.~Albarghouthi, I.~Dillig, and A.~Gurfinkel.
\newblock Maximal specification synthesis.
\newblock In \emph{POPL}, 2016.

\bibitem[Alur et~al.(2013)Alur, Bod{\'{\i}}k, Juniwal, Martin, Raghothaman,
  Seshia, Singh, Solar{-}Lezama, Torlak, and Udupa]{AlurBJMRSSSTU13}
R.~Alur, R.~Bod{\'{\i}}k, G.~Juniwal, M.~M.~K. Martin, M.~Raghothaman, S.~A.
  Seshia, R.~Singh, A.~Solar{-}Lezama, E.~Torlak, and A.~Udupa.
\newblock Syntax-guided synthesis.
\newblock In \emph{FMCAD}, 2013.

\bibitem[Alur et~al.(2015)Alur, {\v C}ern{\'{y}}, and Radhakrishna]{AlurCR15}
R.~Alur, P.~{\v C}ern{\'{y}}, and A.~Radhakrishna.
\newblock Synthesis through unification.
\newblock In \emph{CAV}, 2015.

\bibitem[Augustsson(2014)]{Augustsson14}
L.~Augustsson.
\newblock djinn, the official {Haskell} package webpage.
\newblock http://hackage.haskell.org/package/djinn, 2014.

\bibitem[Brady(2013)]{Brady13}
E.~Brady.
\newblock Idris, a general-purpose dependently typed programming language:
  Design and implementation.
\newblock \emph{J. Funct. Program.}, 23\penalty0 (5):\penalty0 552--593, 2013.

\bibitem[Claessen et~al.(2012)Claessen, Johansson, Ros{\'{e}}n, and
  Smallbone]{ClaessenJRS12}
K.~Claessen, M.~Johansson, D.~Ros{\'{e}}n, and N.~Smallbone.
\newblock Hipspec: Automating inductive proofs of program properties.
\newblock In \emph{ATx/WInG}, 2012.

\bibitem[Coquand(1996)]{Coquand96}
T.~Coquand.
\newblock An algorithm for type-checking dependent types.
\newblock \emph{Sci. Comput. Program.}, 26\penalty0 (1-3):\penalty0 167--177,
  1996.

\bibitem[Davies and Pfenning(2000)]{DaviesPf00}
R.~Davies and F.~Pfenning.
\newblock Intersection types and computational effects.
\newblock In \emph{ICFP}, 2000.

\bibitem[Dillig et~al.(2013)Dillig, Dillig, Li, and McMillan]{DilligDLM13}
I.~Dillig, T.~Dillig, B.~Li, and K.~L. McMillan.
\newblock Inductive invariant generation via abductive inference.
\newblock In \emph{OOPSLA}, 2013.

\bibitem[Dunfield and Pfenning(2004)]{DunfieldPf04}
J.~Dunfield and F.~Pfenning.
\newblock Tridirectional typechecking.
\newblock In \emph{POPL}, 2004.

\bibitem[Feser et~al.(2015)Feser, Chaudhuri, and Dillig]{FeserChDi15}
J.~K. Feser, S.~Chaudhuri, and I.~Dillig.
\newblock Synthesizing data structure transformations from input-output
  examples.
\newblock In \emph{PLDI}, 2015.

\bibitem[Flanagan(2006)]{Flanagan06}
C.~Flanagan.
\newblock Hybrid type checking.
\newblock In \emph{POPL}, 2006.

\bibitem[Frankle et~al.(2016)Frankle, Osera, Walker, and
  Zdancewic]{FrankleOWZ16}
J.~Frankle, P.~Osera, D.~Walker, and S.~Zdancewic.
\newblock Example-directed synthesis: a type-theoretic interpretation.
\newblock In \emph{POPL}, 2016.

\bibitem[Gvero et~al.(2013)Gvero, Kuncak, Kuraj, and Piskac]{GveroKuKuPi13}
T.~Gvero, V.~Kuncak, I.~Kuraj, and R.~Piskac.
\newblock Complete completion using types and weights.
\newblock In \emph{PLDI}, 2013.

\bibitem[Heras et~al.(2013)Heras, Komendantskaya, Johansson, and
  Maclean]{HerasKJM13}
J.~Heras, E.~Komendantskaya, M.~Johansson, and E.~Maclean.
\newblock Proof-pattern recognition and lemma discovery in {ACL2}.
\newblock In \emph{LPAR}, 2013.

\bibitem[Inala et~al.(2015)Inala, Qiu, Lerner, and Solar{-}Lezama]{InalaQLS15}
J.~P. Inala, X.~Qiu, B.~Lerner, and A.~Solar{-}Lezama.
\newblock Type assisted synthesis of recursive transformers on algebraic data
  types.
\newblock \emph{CoRR}, abs/1507.05527, 2015.

\bibitem[Jhala et~al.(2015)Jhala, Seidel, and Vazou]{LiquidHaskellTutorial}
R.~Jhala, E.~Seidel, and N.~Vazou.
\newblock Programming with refinement types (an introduction to liquidhaskell).
\newblock \url{https://ucsd-progsys.github.io/liquidhaskell-tutorial}, 2015.

\bibitem[Kawaguchi et~al.(2009)Kawaguchi, Rondon, and Jhala]{KawaguchiRoJh09}
M.~Kawaguchi, P.~M. Rondon, and R.~Jhala.
\newblock Type-based data structure verification.
\newblock In \emph{PLDI}, 2009.

\bibitem[Kneuss et~al.(2013)Kneuss, Kuraj, Kuncak, and Suter]{KneussKuKuSu13}
E.~Kneuss, I.~Kuraj, V.~Kuncak, and P.~Suter.
\newblock Synthesis modulo recursive functions.
\newblock In \emph{OOPSLA}, 2013.

\bibitem[Leino and Milicevic(2012)]{LeinoMi12}
K.~R.~M. Leino and A.~Milicevic.
\newblock Program extrapolation with {Jennisys}.
\newblock In \emph{OOPSLA}, 2012.

\bibitem[Liffiton et~al.(2015)Liffiton, Previti, Malik, and
  Marques-Silva]{LiffitonPrMaMa15}
M.~Liffiton, A.~Previti, A.~Malik, and J.~Marques-Silva.
\newblock Fast, flexible mus enumeration.
\newblock \emph{Constraints}, pages 1--28, 2015.

\bibitem[Mandelin et~al.(2005)Mandelin, Xu, Bod\'{\i}k, and
  Kimelman]{MandelinXuBoKi05}
D.~Mandelin, L.~Xu, R.~Bod\'{\i}k, and D.~Kimelman.
\newblock Jungloid mining: Helping to navigate the api jungle.
\newblock In \emph{PLDI}, 2005.

\bibitem[Milicevic et~al.(2015)Milicevic, Near, Kang, and
  Jackson]{MilicevicNKJ15}
A.~Milicevic, J.~P. Near, E.~Kang, and D.~Jackson.
\newblock Alloy*: {A} general-purpose higher-order relational constraint
  solver.
\newblock In \emph{ICSE}, 2015.

\bibitem[Montano{-}Rivas et~al.(2012)Montano{-}Rivas, McCasland, Dixon, and
  Bundy]{Montano-RivasMDB12}
O.~Montano{-}Rivas, R.~L. McCasland, L.~Dixon, and A.~Bundy.
\newblock Scheme-based theorem discovery and concept invention.
\newblock \emph{Expert Syst. Appl.}, 39\penalty0 (2):\penalty0 1637--1646,
  2012.

\bibitem[Norell(2009)]{Norell09}
U.~Norell.
\newblock Dependently typed programming in agda.
\newblock In \emph{AFP}, 2009.

\bibitem[Osera and Zdancewic(2015)]{OseraZd15}
P.~Osera and S.~Zdancewic.
\newblock Type-and-example-directed program synthesis.
\newblock In \emph{PLDI}, 2015.

\bibitem[Perelman et~al.(2012)Perelman, Gulwani, Ball, and
  Grossman]{PerelmanGuBaGr12}
D.~Perelman, S.~Gulwani, T.~Ball, and D.~Grossman.
\newblock Type-directed completion of partial expressions.
\newblock In \emph{PLDI}, 2012.

\bibitem[Pierce(2002)]{Pierce02}
B.~C. Pierce.
\newblock \emph{Types and Programming Languages}.
\newblock MIT Press, 2002.

\bibitem[Pierce and Turner(2000)]{PierceTu00}
B.~C. Pierce and D.~N. Turner.
\newblock Local type inference.
\newblock \emph{{ACM} Trans. Program. Lang. Syst.}, 22\penalty0 (1):\penalty0
  1--44, 2000.

\bibitem[Polikarpova and Kuraj(2015)]{SynquidRepo}
N.~Polikarpova and I.~Kuraj.
\newblock Synquid code repository.
\newblock \url{https://bitbucket.org/nadiapolikarpova/synquid/}, 2015.

\bibitem[Rondon et~al.(2008)Rondon, Kawaguchi, and Jhala]{RondonKaJh08}
P.~M. Rondon, M.~Kawaguchi, and R.~Jhala.
\newblock Liquid types.
\newblock In \emph{PLDI}, 2008.

\bibitem[Srivastava and Gulwani(2009)]{SrivastavaGu09}
S.~Srivastava and S.~Gulwani.
\newblock Program verification using templates over predicate abstraction.
\newblock In \emph{PLDI}, 2009.

\bibitem[Suter et~al.(2011)Suter, K{\"{o}}ksal, and Kuncak]{SuterKoKu11}
P.~Suter, A.~S. K{\"{o}}ksal, and V.~Kuncak.
\newblock Satisfiability modulo recursive programs.
\newblock In \emph{SAS}, 2011.

\bibitem[Vazou et~al.(2013)Vazou, Rondon, and Jhala]{VazouRoJh13}
N.~Vazou, P.~M. Rondon, and R.~Jhala.
\newblock Abstract refinement types.
\newblock In \emph{ESOP}, 2013.

\bibitem[Vazou et~al.(2014{\natexlab{a}})Vazou, Seidel, and Jhala]{VazouSeJh14}
N.~Vazou, E.~L. Seidel, and R.~Jhala.
\newblock Liquidhaskell: experience with refinement types in the real world.
\newblock In \emph{Haskell}, 2014{\natexlab{a}}.

\bibitem[Vazou et~al.(2014{\natexlab{b}})Vazou, Seidel, Jhala, Vytiniotis, and
  Jones]{VazouSeJhVyJo14}
N.~Vazou, E.~L. Seidel, R.~Jhala, D.~Vytiniotis, and S.~L.~P. Jones.
\newblock Refinement types for haskell.
\newblock In \emph{ICFP}, 2014{\natexlab{b}}.

\bibitem[Vazou et~al.(2015)Vazou, Bakst, and Jhala]{VazouBJ15}
N.~Vazou, A.~Bakst, and R.~Jhala.
\newblock Bounded refinement types.
\newblock In \emph{ICFP}, 2015.

\bibitem[Wadler(1989)]{Wadler89}
P.~Wadler.
\newblock Theorems for free!
\newblock In \emph{FPCA}, 1989.

\bibitem[Xi and Pfenning(1999)]{XiPf99}
H.~Xi and F.~Pfenning.
\newblock Dependent types in practical programming.
\newblock In \emph{POPL}, 1999.

\end{thebibliography}

\iflong
\appendix
\section{Soundness and Completeness of Round-trip Type Checking}\label{app:proofs}

We show properties of round-trip type checking (\autoref{fig:round-trip})
relative to the purely bottom-up liquid type inference (\autoref{fig:bottom-up}),
using a bidirectional type checking system (\autoref{fig:bidir}) as an intermediate step.

\subsection{Soundness of round-trip type checking}

\addtocounter{theorem}{-2}

\begin{theorem}[Soundness of Round-Trip Type Checking]\label{thm:soundness}
If $\entailsQ \env$, $\env\entailsQ S$, and $\env \entailsQ t \consume S$, then $\env\entailsQ t :: S'$ and $\env \vdash S' \Subt S$.
\end{theorem}

We prove the theorem in two steps.
In step 1, assuming $\env \entailsQ t \consume S$,
we build a derivation of $\env \entailsQ t \consumeB S$
in the bidirectional type system given in \autoref{fig:bidir}.
The main difference between the two derivations is that 
the bidirectional system performs a subtyping check on the boundary of the two directions,
while the round-trip system performs subtyping checks more locally.
To show that the latter subsumes the former,
we prove that whenever the round-trip system concludes $\env \entailsQ t \consume T \produce \contT{C}{T'}$, 
then $\env;C\vdash T'\Subt T$.
In step 2, we build a derivation of $\env\entailsQ t :: S'$ from $\env \entailsQ t \consumeB S$.

\begin{lemma}[Type Strengthening]\label{lemma:strengthen}
If $\env \entailsQ t \consume T \produce \contT{C}{T'}$ then $\env;C\vdash T'\Subt T$.
\end{lemma}
\begin{proof}
By induction on the structure of the derivation. 
Cases \textsc{VarSc}, \textsc{Var$\forall$}, and \textsc{AppFO} are trivial,
since the subtyping check is one of the premises.
For \textsc{AppHO}, from the first premise by induction hypothesis we have
$\env;C\vdash(T_x'\to T')\Subt (\tbot\to T)$;
by definition of function subtyping (\textsc{$\Subt$-Fun} in \autoref{fig:wf-subt})
we get $\env;C\vdash T'\Subt T$.
\end{proof}

\begin{figure}
\small
\textbf{Type Inference (Contextual Types)}\quad$\boxed{\env \entailsQ e :: \hat{T}}$
\begin{gather*}
\inference[\textsc{VarSc}]
{\env(x) =\{B \mid \psi\}}
{\env \entailsQ x :: \{B\mid\nu = x\} }
\\
\inference[\textsc{Var$\forall$}]
{\env(x)=\forall\alpha_i.T'\ &  \env \entailsQ T_i}
{\env \entailsQ x :: [T_i/\alpha]T'}
\\
\inference[\textsc{App}]
{\env \entailsQ e :: \contT{C_1}{(\funT{x}{T_x}{T})}    \\  
\env;C_1 \entailsQ t :: \contT{C_2}{T_x'} \\
\env;C_1;C_2 \vdash T_x' \Subt T_x}
{\env \entailsQ \App{e}{t} :: \contT{C_1;C_2;x\colon T_x'}{T} }
\end{gather*}

\textbf{Type Inference (Context-Free Types)}\quad$\boxed{\env \entailsQ t :: S}$
\begin{gather*}
\inference[\textsc{Subt}]
{\env\entailsQ e :: \contT{C}{T'} & \env;C \vdash T' \Subt T}
{\env \entailsQ e :: T}
\\
\inference[\textsc{Abs}]
{\env\entailsQ(\funT{x}{T_x}{T}) &  \env;x\colon T_x \entailsQ t :: T}
{\env \entailsQ \lambda x.t :: (\funT{x}{T_x}{T})}
\\
\inference[\textsc{If}]
{\env \entailsQ e :: \contT{C}{\{\T{Bool}\mid \psi\}} & \env\entailsQ T \\
\env;C;[\top/\nu]\psi \entailsQ t_1 :: T  &  
\env;C;[\bot/\nu]\psi \entailsQ t_2 :: T}
{\env \entailsQ\T{if}\ e\ \T{then}\ t_1\ \T{else}\ t_2 :: T}
\\
\inference[\textsc{Match}]
{\env \entailsQ e :: \contT{C}{\{D\ T_k\mid\psi\}}  \\
\T{C}_i = T_i^j \to \{D\ T_k\mid\psi'_i\}    &
\env_i = \{x_i^j\colon T_i^j\};[x'/\nu]\psi'_i \\
\env\entailsQ T & \env;C;[x'/\nu]\psi;\env_i \entailsQ t_i :: T}
{\env \entailsQ \T{match}\ e\ \T{with}\ |_i\ \T{C}_i \langle x_i^j\rangle \mapsto t_i :: T}
\\
\inference[\textsc{TAbs}]
{\env \entailsQ t :: T  &  \alpha_i\ \text{not free in}\ \env}
{\env \entailsQ t :: \forall\alpha_i.T}
\\
\inference[\textsc{Fix}]
{\env\entailsQ S & \env;x\colon S \entailsQ t :: S}
{\env \entailsQ \T{fix}\ x.t :: S}
\end{gather*}
\caption{Rules of liquid type inference.}\label{fig:bottom-up}
\end{figure}

\begin{figure}
\small
\textbf{Type Inference}\quad$\boxed{\env \entailsQ e \produceB \hat{T}}$
\begin{gather*}
\inference[\textsc{VarSc}]
{\env(x) =\{B \mid \psi\}}
{\env \entailsQ x \produceB \{B\mid\nu = x\} }
\\
\inference[\textsc{Var$\forall$}]
{\env(x)=\forall\alpha_i.T'\ &  \env \entailsQ T_i}
{\env \entailsQ x \produceB [T_i/\alpha]T'}
\\
\inference[\textsc{AppFO}]
{\env \entailsQ e_1 \produceB \contT{C_1}{(\funT{x}{\{B \mid \psi\}}{T})}    \\  
\env;C_1 \entailsQ e_2\produceB \contT{C_2}{T_x} \\
\env;C_1;C_2 \vdash T_x \Subt \{B \mid \psi\}}
{\env \entailsQ \App{e_1}{e_2} \produceB \contT{C_1;C_2;x\colon T_x}{T} }
\\
\inference[\textsc{AppHO}]
{\env \entailsQ e \produceB \contT{C}{(T_x \to T)}   \\  
\env;C \entailsQ f \consumeB T_x}
{\env \entailsQ \App{e}{f} \produceB \contT{C}{T}}
\end{gather*}

\textbf{Type Checking}\quad$\boxed{\env \entailsQ t \consumeB S}$
\begin{gather*}
\inference[\textsc{IE}]
{\env \entailsQ e \produceB \contT{C}{T'} & \env;C \vdash T' \Subt T}
{\env \entailsQ e \consumeB T}
\\
\inference[\textsc{Abs}]
{\env;y\colon T_x \entailsQ t \consumeB [y/x]T}
{\env \entailsQ \lambda y.t \consumeB (\funT{x}{T_x}{T})}
\\
\inference[\textsc{If}]
{\env \entailsQ e \produceB \contT{C}{\{\T{Bool}\mid \psi\}} \\
\env;C;[\top/\nu]\psi \entailsQ t_1 \consumeB T  &  
\env;C;[\bot/\nu]\psi \entailsQ t_2 \consumeB T}
{\env \entailsQ\T{if}\ e\ \T{then}\ t_1\ \T{else}\ t_2 \consumeB T}
\\
\inference[\textsc{Match}]
{\env \entailsQ e \produceB \contT{C}{\{D\ T_k\mid\psi\}}  \\
\T{C}_i = T_i^j \to \{D\ T_k\mid\psi'_i\}    &
\env_i = \{x_i^j\colon T_i^j\};[x'/\nu]\psi'_i \\
\env;C;[x'/\nu]\psi;\env_i \entailsQ t_i \consumeB T}
{\env \entailsQ \T{match}\ e\ \T{with}\ |_i\ \T{C}_i \langle x_i^j\rangle \mapsto t_i \consumeB T}
\\
\inference[\textsc{TAbs}]
{\env \entailsQ t \consumeB T  &  \alpha_i\ \text{not free in}\ \env}
{\env \entailsQ t \consumeB \forall\alpha_i.T}
\\
\inference[\textsc{Fix}]
{\env;x\colon S^\before \entailsQ t \consumeB S}
{\env \entailsQ \T{fix}\ x.t \consumeB S}
\end{gather*}
\caption{Rules of bidirectional type checking.}\label{fig:bidir}
\end{figure}

\begin{lemma}[Round-trip to Bidirectional]\label{lemma:rt-to-bidir}
If $\env \entailsQ t \consume S$ then $\env \entailsQ t \consumeB S$;
additionally, if $t$ is an E-term and $\env \entailsQ t \consume T \produce \hat{T'}$,
then $\env \entailsQ t \produceB \hat{T'}$.
\end{lemma}
Note that there is a one-to-one correspondence between the rules of the two systems (given by their names);
our construction always uses corresponding rules.
\begin{proof}
By induction on the structure of the derivation.
Let us first assume an E-term $e$ and a derivation of $\env \entailsQ e \consume T \produce \hat{T'}$;
we will build a derivation of $\env \entailsQ e \produceB \hat{T'}$.
Cases \textsc{VarSc} and \textsc{Var$\forall$} are trivial
since the bidirectional premises for these rules are a subset of the round-trip premises.
Rule \textsc{AppHO} has the same structure in the two systems,
thus we invoke the induction hypothesis in a straightforward manner.
Finally, in the rule \textsc{AppFO}, we obtain the first two premises from the induction hypothesis;
the third premise, $\env;C_1;C_2 \vdash T_x \Subt \{B \mid \psi\}$,
follows from the round-trip premise $\env;C_1 \entailsQ e_2 \consume \{B \mid \psi\} \produce \contT{C_2}{T_x}$
and Lemma~\ref{lemma:strengthen}.

We are left to deal with the checking rules.
For all checking rules except \textsc{IE} we can obtain the required bidirectional premises
in a straightforward manner from the induction hypothesis.
For \textsc{IE}, we obtain the first premise from the induction hypothesis;
the second premise, $\env;C \vdash T' \Subt T$,
follows from the round-trip premise $\env \entailsQ e \consume T \produce \contT{C}{T'}$
and Lemma~\ref{lemma:strengthen}.
\end{proof}

\begin{lemma}[Bidirectional to Bottom-up]\label{lemma:bidir-to-bu}
If $\entailsQ \env$, $\env\entailsQ S$, and $\env \entailsQ t \consumeB S$, then $\env\entailsQ t :: S'$ and $\env \vdash S' \Subt S$;
additionally if $t$ is an E-term and $\env \entailsQ t \produceB \hat{T}$ then $\env\entailsQ t :: \hat{T}$.
\end{lemma}
\begin{proof}
By induction on the structure of the derivation.
Let us first assume an E-term $e$ and a derivation of $\env \entailsQ e \produceB \hat{T}$;
we will build a derivation of $\env \entailsQ e :: \hat{T}$.
Cases \textsc{VarSc}, \textsc{Var$\forall$}, and \textsc{AppFO} are trivial
(in the latter case, \textsc{AppFO} is replaced by \textsc{App}).
For the \textsc{AppHO} case, we also construct a derivation using \textsc{App} as the root;
from the second premise of \textsc{AppHO}, $\env;C \entailsQ f \consumeB T_x$, 
by induction hypothesis we obtain $T_x'$ such that $\env;C \entailsQ f :: T_x'$ and $\env;C \vdash T_x' \Subt T_x$.
Thus, we obtain the two missing premises of \textsc{App} (with $C_2$ empty).
Note that to apply the induction hypothesis, we need to show that the argument type $T_x$ is liquid ($\env\entailsQ T_x$);
this follows from the assumption $\entailsQ \env$ (since $T_x$ is a precondition of a function from $\env$).

For rule \textsc{IE} it is straightforward to obtain the inferred type $T'$ from the induction hypothesis.
For rule \textsc{Abs}, from the derivation of $\env \entailsQ \lambda x.t \consumeB (\funT{x}{T_x}{T})$,
we construct the derivation of $\env \entailsQ \lambda x.t :: (\funT{x}{T_x}{T})$ (\ie we pick the bidirectional goal type as the inferred type);
this is possible since by the assumptions of the lemma the type is liquid: $\env\entailsQ \funT{x}{T_x}{T}$,
which satisfies the first premise of the bottom-up rule.
To satisfy the second premise,
from the induction hypothesis we infer $\env;x:T_x \entailsQ t :: T'$, where $\env;x:T_x\vdash T'\Subt T$,
and then using the \textsc{Subt} rule we get $\env;x:T_x \entailsQ t :: T$.
To apply the induction hypothesis, we need to show that $T$ is liquid,
which follows from the assumption that $\funT{x}{T_x}{T}$ is liquid and the definition of liquid function types.
The cases \textsc{If}, \textsc{Match}, \textsc{TAbs} are analogous to \textsc{Abs}.
For \textsc{Fix} we additionally rely on the fact that $\env\vdash S\Subt S^\before$;
this follows from the definition of termination weakening, 
which only adds refinements to preconditions inside $S$.
\end{proof}

\begin{proof}[Proof of Theorem \ref{thm:soundness}]
Straightforward by combining Lemma~\ref{lemma:rt-to-bidir} and Lemma~\ref{lemma:bidir-to-bu}.
\end{proof}

\subsection{Completeness of round-trip type checking}

As noted in \autoref{sec:theory:proofs},
we show completeness for a termination-oblivious variant of the round-trip system,
obtained from \autoref{fig:round-trip} by replacing $S^{\before}$ in the premise of the \textsc{Fix} rule by $S$.
We denote the checking judgment of the modified system as $\env \entailsQ^* t \consume S$.

\begin{theorem}[Completeness of round-trip type checking]\label{thm:completeness}
If $\env\entailsQ t :: S$, then $\env \entailsQ^* t \consume S$.
\end{theorem}

Like in the case of soundness, we do the proof in two steps.
We first go from a bottom-up derivation to a bidirectional derivation,
essentially by replacing all inference judgments for I-terms with checking judgments. 
Going from a bidirectional derivation to a round-trip one is more complex:
it amounts to proving that the ``early'' checks of the round-trip system%
---that is, subtyping checks on incomplete applications---%
only reject terms that would be later rejected by the bidirectional \textsc{IE} rule.

\begin{lemma}[Bottom-up to Bidirectional]\label{lemma:bu-to-bidir}
If $\env\entailsQ t :: S$, then $\env \entailsQ^* t \consumeB S$;
additionally if $t$ is an E-term and $\env\entailsQ t :: \hat{T}$, then $\env \entailsQ^* t \produceB \hat{T}$.
\end{lemma}
\begin{proof}
By induction on the structure of the derivation.
Let us first assume an E-term $e$ and a derivation of $\env \entailsQ e :: \hat{T}$;
we will build a derivation of $\env \entailsQ^* e \produceB \hat{T}$.
Cases \textsc{VarSc} and \textsc{Var$\forall$} are trivial.
For \textsc{App}, we have to consider two cases,
depending on whether the argument is an E-term or a function term
(remember, branching terms are disallowed on right-hand sides of applications).
In the former case, we can build a derivation with \textsc{AppFO} in a straightforward manner.
In the latter case, note that the context $C_2$ has to be empty,
because the only rule that could have inferred the type of $t$ is \textsc{Abs},
and that rule does not generate contextual types.
With that observation, we can build a derivation with \textsc{AppHO}: 
we get $\env;C_1\entailsQ t :: T_x'$ and $\env;C_1\vdash T_x'\Subt T_x$,
yielding $\env;C_1\entailsQ t :: T_x$ by \textsc{Subt},
which gives us the required second premise of \textsc{AppHO} by the induction hypothesis.

For all the context-free inference rules,
it is straightforward to construct a bidirectional derivation
using the corresponding checking rule,
since the premises of these rules are a subset of the premises of their bottom-up counterparts.
\end{proof}

\begin{lemma}[Bidirectional to Round-Trip]\label{lemma:bidir-to-rt}
If $\env \entailsQ t \consumeB S$ then $\env \entailsQ t \consume S$;
additionally if $t$ is an E-term and $\env\entailsQ t \produceB \contT{C}{T}$ 
and there exists a type $U$ such that $\env\vdash U$ and $\env;C\vdash T\Subt U$
then $\env \entailsQ t \consume U \produce \contT{C}{T}$.
\end{lemma}
\begin{proof}
By induction on the structure of the derivation.
Let us first assume an E-term $e$ and a type $U$, 
such that $\env \entailsQ e \produceB \contT{C}{T}$, $\env\vdash U$, and $\env;C\vdash T\Subt U$;
we will build a derivation of $\env \entailsQ e \consume U \produce \contT{C}{T}$.
Cases \textsc{VarSc} and \textsc{Var$\forall$} are trivial.

Consider the rule \textsc{AppFO}:
to get the first premise of the round-trip version, we apply the induction hypothesis,
picking $\{B\mid \bot\}\to U$ as the goal;
it is easy to show that this type is well-formed in $\env$ (since so is $U$);
the challenging part is to show $\env;C_1\vdash (\funT{x}{\{B\mid \psi\}}{T}) \Subt (\funT{x}{\{B\mid \bot\}}{U})$.
By the definition of function subtyping,
this simplifies to $\env;C_1\vdash \{B\mid \bot\} \Subt \{B\mid \psi\}$ (trivial),
and $\env;C_1;x:\{B\mid \bot\}\vdash T\Subt U$.
We know by the assumption the lemma makes about $U$ that $\env;C_1;C_2;x:T_x\vdash T\Subt U$,
so the only thing left to show is that 
$\llbracket\env;C_1;x:\{B\mid \bot\}\rrbracket_{T\Subt U} \Implies \llbracket\env;C_1;C_2;x:T_x\rrbracket_{T\Subt U}$
(note that $C_2$ only contains bindings for variables that appear inside $T_x$ and nowhere else).
We consider two cases: if $x$ does not actually appear in $T$,
then both formulas above are equivalent to $\llbracket\env;C_1\rrbracket_{T\Subt U}$, and thus the implication holds;
otherwise, one of the conjuncts in the left-hand side is $\bot$, and thus the implication also holds.

To get the second second premise of \textsc{AppFO},
we pick $\{B\mid \psi\}$ (the precondition from the inferred type of $e_1$) as the goal.
It is easy to show that this type is well-formed in $\env;C_1$; 
and moreover $\env;C_1;C_2\vdash T_x \Subt \{B\mid \psi\}$ by the third premise of the bidirectional rule.
The third premise of \textsc{AppFO} follows directly from the lemma's assumption about $U$.

Consider the rule \textsc{AppHO}.
Obtaining its first premise is similar to the case of \textsc{AppFO},
except it is easier to show that $\env;C\vdash T_x\to T \Subt \tbot\to U$,
since these function types are not dependent and thus the above is equivalent to $\env;C\vdash T \Subt U$
(which we get directly from the lemma's assumption about $U$).
The second premise is obtained directly from the induction hypothesis.

For the rule \textsc{IE}, 
we construct the required premise $\env \entailsQ e \consume T \produce \contT{C}{T'}$
from the bidirectional derivation for $\env \entailsQ e \produceB \contT{C}{T'}$,
taking into account that $\env;C\vdash T'\Subt T$ is required by the bidirectional rule,
and thus $T$ is a valid goal type for the round-trip derivation. 

The rest of the rules are trivial, since they have the same shape in both systems;
the only two exceptions are the first premises of \textsc{If} and \textsc{Match},
which require a strengthening judgment;
in both cases it is straightforward to show that the respective goal types (\T{Bool} and $\ttop$)
are supertypes of the type inferred by the bidirectional system.
\end{proof}

\begin{proof}[Proof of Theorem \ref{thm:completeness}]
Straightforward by combining Lemma~\ref{lemma:bu-to-bidir} and Lemma~\ref{lemma:bidir-to-rt}.
\end{proof}

\end{document}